%% file: main.tex
\documentclass[fleqn,10pt]{wlscirep}
\usepackage[utf8]{inputenc}
\usepackage[T1]{fontenc}

\usepackage{xr-hyper}
\usepackage{hyperref}
\newcommand*\sref{%
	Appendix \ref}

\usepackage{xcolor}

\title{Scaling of contact networks for epidemic spreading in urban transit systems}
\author[a]{Xinwu Qian}
\author[b]{Lijun Sun}
\author[a,1]{Satish V. Ukkusuri}

\affil[a]{Lyles School of Civil Engineering, Purdue University}
\affil[b]{Department of Civil Engineering, McGill University}

\begin{abstract}
Improved mobility not only contributes to more intensive human activities but also facilitates the spread of communicable disease, thus constituting a major threat to billions of urban commuters. In this study, we present a multi-city investigation of communicable diseases percolating among metro travelers. We use smart card data from three megacities in China to construct individual-level contact networks, based on which the spread of disease is modeled and studied. We observe that, though differing in urban forms, network layouts, and mobility patterns, the metro systems of the three cities share similar contact network structures. This motivates us to develop a universal generation model that captures the distributions of the number of contacts as well as the contact duration among individual travelers. This model explains how the structural properties of the metro contact network are associated with the risk level of communicable diseases. Our results highlight the vulnerability of urban mass transit systems during disease outbreaks and suggest important planning and operation strategies for mitigating the risk of communicable diseases.
\end{abstract}

\begin{document}

\flushbottom
\maketitle
%
%
\thispagestyle{empty}

\input{introduction.tex}

\input{metro_contact_network.tex}

\input{disease_dynamics.tex}

\input{generation_model.tex}

\input{discussion.tex}

\bibliography{sample}

\section*{Author contributions statement}
X.Q. and S.V.U. designed the research; X.Q. and L.S. performed research; X.Q. and L.S. analyzed data; X.Q., L.S. and S.V.U. wrote the paper.

\section*{Additional information}

\textbf{Competing interests} The authors declare no competing interests.


\end{document}


\maketitle
\selectfont

%
%


\section{Data}
\label{SI:data}
\subsection{Metro Trip data}
The metro smart card transaction data are from three major cities in China: Shanghai, Guangzhou, and Shenzhen. These data have similar structure, with each record containing the information of smart card ID, transaction ID, transaction time, boarding station/time, and alighting station/time. The transaction type indicates if the transaction is entry or exit at the transaction station. Since each smart card is associated with a unique ID, we can therefore construct the trip sequences for each commuter (each card) based on the transaction time, transaction type and the location. We present a sample of smart card transaction data of Shenzhen on April 21st, 2016 in Tab~\ref{SI_tab:data_records}.
\begin{table}[!h]
\centering
\caption{Sample records of metro transaction data of Shenzhen (2016-04-21). For Transaction Type, 21 indicates that the traveler left the system and 22 denotes the entry of the traveler.}
\begin{tabular}{p{2cm}p{2cm}p{2cm}p{1cm}p{1cm}}
	\toprule
User ID & Transaction Type & Time & Station ID  \\
\hline
80357781&22&08:39:50&1\\
290452424&22&08:39:32&1\\
20353676&22&09:41:43&1\\
361341888&21&07:15:36&1\\
329838057&22&07:47:08&1\\
667519928&22&08:34:07&1\\
329213920&21&07:37:19&1\\
\bottomrule
\end{tabular}
\label{SI_tab:data_records}
\end{table}

For each city, we extract the data of five work days for further analysis. A summary statistics of the data and the size of the metro networks that corresponded to the period of available data is shown in Table~\ref{tab:metro_stat}. When compared with official statistics of daily ridership, we observe that the smart card transaction data may cover over 60\% of total daily travelers and can well reflect the trip dynamics of regular metro users. The metro networks in these cities have distinct layouts which are tailored to the urban form. Shanghai metro is the metro system with the longest total mileage and largest number of stations. It also has the highest number of daily travelers. Guangzhou and Shenzhen are similar in terms of the size of the metro networks, however, the shape of the metro network differs. In particular, Shenzhen is a stripe-shape city where commercial areas are located in the middle and residential places are distributed at east and west sides of the city. The layouts of metro networks and half-hourly passenger demand distributions of the three cities are presented in Figure~\ref{fig:metro_networks}. Note that the layouts presented here correspond to the period of time when the data were collected. 
\begin{figure}[!ht]
    \centering
    \subfloat[Guangzhou metro network]{\includegraphics[width=50mm,height=30mm]{fig/gz_metro.png}}
    \subfloat[Shanghai metro network]{\includegraphics[width=50mm,height=30mm]{fig/sh_metro.png}}
    \subfloat[Shenzhen metro network]{\includegraphics[width=50mm,height=30mm]{fig/sz_metro.png}}\\
    \subfloat[Passenger demand]{\includegraphics[width=60mm]{fig/demand_threecities.pdf}}
    \caption{Metro network layouts of the three cities and the half-hourly passenger demand distributions.}
    \label{fig:metro_networks}
\end{figure}

\begin{table}[!ht]
\centering
\caption{Summary of metro card transaction data from three major cities in China}
\begin{tabular}{p{2cm}|p{2cm}|p{2cm}|p{2cm}|p{2cm}|p{3cm}}
	\toprule
City      & Start date & End date   & \# metro lines & \# stations & Average daily records \\
\hline
Guangzhou & 2017.04.13           & 2017.04.17           & 8                     & 166                & 1.6 million           \\
Shanghai  & 2015.04.13 & 2015.04.17 & 13                    & 288                & 4.16 million       \\
Shenzhen  & 2016.04.14 & 2016.04.18 & 5                     & 118                & 2.13 million          \\
\bottomrule
\end{tabular}
\label{tab:metro_stat}
\end{table}

\subsection{Operation data}
In order to infer the contact among travelers, we also need operation data which include the trip time between two adjacent metro stations, the approximate transfer time at transfer stations, and the frequency of metro trains.

To obtain these data, we developed web crawlers and extracted the metro station adjacency matrix from GaoDe Map API~\cite{gaodemap} as the representation of the metro system layouts . In addition, the time tables of the three metro systems were obtained from their official websites~\cite{guangzhou_timetable,shanghai_timetable,shenzhen_timetable}, which contain the travel time between two stations as well as the frequency of the metro trains during different time periods. Finally, the transfer time required between at the transfer station is calculated by identifying a route that needs a transfer at the station, quoting the travel time of the route using GaoDe Map API and subtracting the travel time of the route based on the values that we obtained from the timetable.

\section{Metro Contact Network}
\label{SI:mcn_generation_algorithm}
\subsection{Unweighted metro contact network}
Based on the smart card transaction data and operation data for metro networks, we next develop the algorithm for constructing metro contact network (MCN). The contact network is constructed at individual level and we consider both unweighted and weighted contact networks. For unweighted MCN, each node represents a traveler and the link between a pair of nodes denotes that the two travelers will have positive probability to be on the same metro train. Since the smart card data only contain the time and location of a traveler entering and leaving the metro system, we need to infer if two travelers will be on the same metro train based on their trip starting time, origin station and destination station. And this further requires us to predict their travel route within the metro system. While the operation timetable of metro system is largely reliable, we assume that all travelers will follow the shortest route between two trip origin and destination (including both station-wise travel time and transfer time). Based on the predicted travel route, we can therefore determine if an link exists between two travelers following
\begin{enumerate}
	\item Find the shortest travel route for each traveler.
	\item For each pair of passenger $i$ and $j$, determine if they have overlapping travel segments $L_{ij}$.
	\item If $L_{ij}\geq 2$, determine their first meeting location, station $m$, and calculate their arrival time at the meeting station $t_{i,m},t_{j,m}$ respectively.
	\item Compute the probability of contact between travelers $i$ and $j$ based on $t_{i,m},t_{j,m}$ and the frequency $f$ of metro lines.
	\item Repeat this process until all pairs of travelers are processed. Output $G$.
\end{enumerate}
In step 1, the shortest route can be computed using the Dijkstra algorithm with the travel time adjacency matrix $A$, and the shortest routes are stored as sequences of the stations $p_i=\{s_1,s_2,..,s_P\}$ along the routes. In step 2 and 3, the overlapping travel segments of two travelers $i$ and $j$ can be identified as the longest common subsequence (LCS) of their routes $p_i$ and $p_j$. In our case, a valid LCS that may grant contact is the LCS of length 2 or higher, indicating that the two travelers share at least one trip segment. In step 4, $t_{i,m}, t{j,m}$ can be computed from their departure time and the trip time between their origin station and first meeting station $m$. Then their contact probability $p_{ij}$ follows
\begin{equation}
p_{ij}=\begin{cases}
1-\frac{|t_{i,m}-t_{j,m}|}{f}, \text{if $|t_{i,m}-t_{j,m}|<f$}\\
0,\,\text{if $|t_{i,m}-t_{j,m}|>=f$}
\end{cases}
\end{equation}
This suggests that two travelers will have positive contact probability if the gap between their arrival time at $m$ is less than the frequency $f$ of metro trains. And this probability decreases linearly considering the uniform arrival of metro trains following frequency $f$.
\subsection{Weighted metro contact network}
Based on unweighted MCNs, we further assign the weight to each link in unweighted MCNs to produce weighted MCNs. In the context of modeling the spread of communicable diseases, the weight on each link has the physical meaning as the \textbf{expected contact duration between two individuals within effective transmission range}. By effective transmission range, we consider that two individuals are close enough so that the airborne transmission of an communicable disease is feasible. This follows from the definition of the effective range for droplet transmission, which is usually less than 3 feet while certain disease such as SARS may reach 6 feet~\cite{siegel20072007}. Let $C$ denote the scaling parameter for the effective transmission range, we have the weight between traveler $i$ and $j$ as:
\begin{equation}
d_{ij}=\frac{p_{ij}L_{ij}}{C}
\label{SI_eq:contact_duration}
\end{equation}
Considering that a metro train consists of 6 coaches, with each coach of length 72 feet, then C may take the value of 144 if the effective transmission range is 3 feet. And equation~\ref{SI_eq:contact_duration} characterizes the expected contact duration as the product of the probability for being in effective transmission range $p_{ij}/C$ and the duration of the contact $L_{ij}$, with the underlying assumption that travelers will uniformly distribute themselves among all metro coaches. The use of $C$ naturally captures the behavior of travelers to avoid congested coaches during travel. With increasing number of travelers (e.g. more number of nodes in MCNs), this also characterizes the linearly increasing chance of close contacts, where the total contact duration of each individual is the row sum of the contact duration matrix. 

Finally, for the transmission rate of communicable disease, let $\beta$ denote the transmission strength per unit time, we have the transmission rate between two travelers as:
\begin{equation}
\beta_{ij}=\beta d_{ij}
\label{SI_eq:strength}
\end{equation}

With the above processes, we can use the smart card data to generate sample unweigted and weighted MCNs. Specifically, the smart card data can be aggregated to generate the distributions for trip origin and destinations and the arrival time at each station. We then sample $N$ travelers following the distributions, where each traveler has their time of arrival, the origin station and the destination station. And the MCNs with $N$ nodes can consequently be constructed following the generation process for MCNs. We denote $A$ as the adjacency matrix of the generated MCNs, with each entry $A_{ij}=d_{ij}$.

\subsection{Structural property of MCNs}
\label{SI:mcn_generation}

\begin{table}[!ht]
\centering
\caption{Summary statistics of the MCNs of various number of nodes for Guangzhou. In the table, $<k>$ represents the average unweighted degree and $<d>$ represents the average weighted degree.  }
\begin{tabular}{p{1cm}p{1cm}p{1cm}p{1cm}p{1cm}lllllll}
	\hline
Number of nodes & Average path length & Average clustering coefficient & \hspace{0pt}Assortativity & \hspace{0pt}Diameter & \textless{}k\textgreater{} & \textless{}$k^2$\textgreater{} & \textless{}$k_{max}$\textgreater{} & \textless{}d\textgreater{} & \textless{}$d^2$\textgreater{} & \textless{}$d_{max}$\textgreater{} \\
\toprule
500             & 2.99                & 0.48                   & 0.27          & 6.90     & 18.00                     & 416.37                                        & 45.10                          & 7.97                       & 102.52                                        & 30.74                          \\
1000            & 2.79                & 0.49                   & 0.26          & 7.20     & 35.09                      & 1,562.10                                      & 89.50                          & 15.95                      & 401.66                                        & 59.99                          \\
1500            & 2.68                & 0.49                   & 0.24          & 6.80     & 53.64                      & 3,610.10                                      & 126.70                         & 23.68                      & 873.08                                        & 89.61                          \\
2000            & 2.64                & 0.49                   & 0.25          & 7.10     & 71.03                      & 6,322.30                                      & 169.70                         & 31.83                      & 1,564.80                                      & 117.92                         \\
2500            & 2.60                & 0.49                   & 0.25          & 6.30     & 89.55                      & 10,027                                     & 209.70                         & 40.22                      & 2,503.40                                      & 148.33                         \\
3000            & 2.58                & 0.49                   & 0.25          & 6.40     & 106.15                     & 14,085                                     & 247.20                         & 47.35                      & 3,434.80                                      & 173.65                         \\
3500            & 2.55                & 0.49                   & 0.25          & 6.60     & 124.88                     & 19,444                                     & 292.20                         & 56.30                      & 4,862.50                                      & 208.17                         \\
4000            & 2.54                & 0.49                   & 0.24          & 6.40     & 141.92                     & 25,075                                     & 337.50                         & 63.51                      & 6,168.50                                      & 237.69                         \\
4500            & 2.52                & 0.49                   & 0.26          & 5.90     & 160.94                     & 32,346                                     & 378.70                         & 72.05                      & 7,966.40                                      & 259.52                         \\
5000            & 2.51                & 0.49                   & 0.24          & 6.40     & 177.68                     & 39,327                                     & 411.20                         & 79.09                      & 9,583.30                                      & 292.39                         \\
5500            & 2.50                & 0.49                   & 0.25          & 6.20     & 197.53                     & 48,585                                     & 461.80                         & 88.56                      & 12,000                                     & 318.63                         \\
6000            & 2.49                & 0.49                   & 0.24          & 5.90     & 212.67                     & 56,242                                     & 491.40                         & 95.82                      & 14,058                                     & 354.27                         \\
6500            & 2.49                & 0.49                   & 0.25          & 6.20     & 231.81                     & 67,051                                     & 541.50                         & 104.12                     & 16,644                                     & 377.04                         \\
7000            & 2.48                & 0.49                   & 0.25          & 5.90     & 249.96                     & 77,836                                     & 579.30                         & 112.71                     & 19,554                                     & 401.65                         \\
7500            & 2.47                & 0.49                   & 0.25          & 5.80     & 267.74                     & 89,419                                     & 626.30                         & 119.68                     & 21,986                                     & 443.35                         \\
8000            & 2.46                & 0.49                   & 0.25          & 5.90     & 285.40                     & 101,360                                    & 661.60                         & 127.38                     & 24,787                                     & 460.05                         \\
8500            & 2.46                & 0.49                   & 0.23          & 6.20     & 302.02                     & 113,320                                    & 688.30                         & 134.64                     & 27,515                                     & 494.58                         \\
9000            & 2.45                & 0.49                   & 0.24          & 5.90     & 321.19                     & 128,210                                    & 738.70                         & 144.12                     & 31,663                                     & 520.86                         \\
9500            & 2.45                & 0.49                   & 0.25          & 6.20     & 339.39                     & 143,330                                    & 776.80                         & 152.16                     & 35,326                                     & 541.30                         \\
10000           & 2.44                & 0.49                   & 0.25          & 6.00     & 356.57                     & 158,060                                    & 819.80                         & 159.84                     & 38,966                                     & 576.08                        \\
\bottomrule
\end{tabular}
\label{SI_tab:MCN_stat_gz}
\end{table}

\begin{table}[!ht]
\centering
\caption{Summary statistics of the MCNs of various number of nodes for Shanghai.}
\begin{tabular}{p{1cm}p{1cm}p{1cm}p{1cm}p{1cm}llllll}
	\toprule
Number of nodes & Average path length & Average clustering coefficient & \hspace{0pt}Assortativity & \hspace{0pt}Diameter & \textless{}k\textgreater{} & \textless{}$k^2$\textgreater{} & \textless{}$k_{max}$\textgreater{} & \textless{}d\textgreater{} & \textless{}$d^2$\textgreater{} & \textless{}$d_{max}$\textgreater{} \\
\toprule
500             & 3.14                & 0.46                   & 0.27          & 7.40     & 14.48                      & 280.46                                        & 41.90                          & 9.50                       & 139.37                                        & 35.67                          \\
1000            & 2.89                & 0.46                   & 0.30          & 7.60     & 29.14                      & 1125.50                                       & 77.20                          & 18.78                      & 528.15                                        & 67.91                          \\
1500            & 2.77                & 0.47                   & 0.29          & 7.80     & 43.56                      & 2483.10                                       & 117.40                         & 28.37                      & 1179.50                                       & 102.35                         \\
2000            & 2.71                & 0.47                   & 0.30          & 7.60     & 58.74                      & 4516.80                                       & 154.30                         & 37.91                      & 2085.60                                       & 136.65                         \\
2500            & 2.67                & 0.47                   & 0.30          & 7.80     & 73.57                      & 7098.20                                       & 192.60                         & 47.48                      & 3275.60                                       & 172.57                         \\
3000            & 2.64                & 0.47                   & 0.30          & 8.30     & 87.83                      & 10030                                      & 231.50                         & 57.26                      & 4737.90                                       & 204.34                         \\
3500            & 2.62                & 0.47                   & 0.29          & 7.50     & 102.58                     & 13642                                      & 274.60                         & 66.58                      & 6386                                       & 240.23                         \\
4000            & 2.60                & 0.47                   & 0.30          & 7.40     & 116.81                     & 17751                                      & 306.90                         & 75.84                      & 8276.20                                       & 261.84                         \\
4500            & 2.59                & 0.47                   & 0.30          & 7.30     & 131.69                     & 22497                                      & 344.40                         & 85.60                      & 10504                                      & 295.58                         \\
5000            & 2.57                & 0.47                   & 0.30          & 7.30     & 146.25                     & 27856                                      & 389.10                         & 95.46                      & 13144                                      & 335.95                         \\
5500            & 2.56                & 0.47                   & 0.30          & 7.40     & 160.56                     & 33407                                      & 423.10                         & 104.59                     & 15692                                      & 374.46                         \\
6000            & 2.55                & 0.47                   & 0.29          & 7.30     & 174.79                     & 39561                                      & 465                         & 113.71                     & 18553                                      & 410.25                         \\
6500            & 2.54                & 0.47                   & 0.29          & 7.30     & 190                     & 46854                                      & 504.10                         & 123.32                     & 21845                                      & 429.61                         \\
7000            & 2.53                & 0.47                   & 0.30          & 7.10     & 205.03                     & 54479                                      & 552.70                         & 133.28                     & 25513                                      & 471.11                         \\
7500            & 2.52                & 0.47                   & 0.29          & 8     & 221.12                     & 63337                                      & 584                         & 143.07                     & 29267                                      & 492.32                         \\
8000            & 2.52                & 0.47                   & 0.29          & 8     & 232.36                     & 69754                                      & 609.50                         & 151.18                     & 32655                                      & 526.39                         \\
8500            & 2.51                & 0.47                   & 0.29          & 6.90     & 251.55                     & 81864                                      & 656.70                         & 164.12                     & 38671                                      & 575.84                         \\
9000            & 2.50                & 0.47                   & 0.29          & 7.20     & 263.96                     & 89996                                      & 698                         & 171.97                     & 42270                                      & 612.77                         \\
9500            & 2.50                & 0.47                   & 0.29          & 7.40     & 277.66                     & 99696                                      & 734.20                         & 180.60                     & 46780                                      & 641.35                         \\
10000           & 2.49                & 0.47                   & 0.30          & 7.10     & 292.74                     & 110970                                     & 778.80                         & 190.79                     & 51981                                      & 676.15              \\
\bottomrule
\end{tabular}
\label{SI_tab:MCN_stat_sh}
\end{table}

\begin{table}[!ht]
\centering
\caption{Summary statistics of the MCNs of various number of nodes for Shenzhen.}
\begin{tabular}{p{1cm}p{1cm}p{1cm}p{1cm}p{1cm}llllll}
	\toprule
Number of nodes & Average path length & Average clustering coefficient & \hspace{0pt}Assortativity & \hspace{0pt}Diameter & \textless{}k\textgreater{} & \textless{}$k^2$\textgreater{} & \textless{}$k_{max}$\textgreater{} & \textless{}d\textgreater{} & \textless{}$d^2$\textgreater{} & \textless{}$d_{max}$\textgreater{} \\
\toprule
500             & 3.01                & 0.53                   & 0.26          & 6.70     & 20.51                      & 552.23                                        & 51.60                          & 11.54                      & 216.49                                        & 42.15                          \\
1000            & 2.78                & 0.54                   & 0.27          & 6     & 41.59                      & 2227.10                                       & 101.20                         & 23.45                      & 874.57                                        & 86.62                          \\
1500            & 2.69                & 0.54                   & 0.25          & 5.20     & 61.70                      & 4882.30                                       & 154.60                         & 35.44                      & 1996.70                                       & 132.55                         \\
2000            & 2.64                & 0.54                   & 0.23          & 5.60     & 83.32                      & 8839                                       & 199.90                         & 47.70                      & 3570.50                                       & 175.15                         \\
2500            & 2.60                & 0.54                   & 0.26          & 5.10     & 104.38                     & 13918                                      & 272.10                         & 60.10                      & 5650.10                                       & 223.27                         \\
3000            & 2.57                & 0.54                   & 0.25          & 5.10     & 125.54                     & 20095                                      & 328.70                         & 71.94                      & 8064                                       & 256.44                         \\
3500            & 2.56                & 0.54                   & 0.24          & 5     & 144.51                     & 26587                                      & 361.90                         & 82.46                      & 10620                                      & 295.50                         \\
4000            & 2.55                & 0.54                   & 0.24          & 5.10     & 166.57                     & 35318                                      & 404.40                         & 95.47                      & 14250                                      & 354.75                         \\
4500            & 2.53                & 0.54                   & 0.24          & 5     & 187.63                     & 44732                                      & 466.60                         & 107.36                     & 17976                                      & 390.15                         \\
5000            & 2.53                & 0.54                   & 0.25          & 5     & 206.49                     & 54322                                      & 522.20                         & 118.88                     & 22117                                      & 437.16                         \\
5500            & 2.51                & 0.54                   & 0.24          & 5     & 228.79                     & 66332                                      & 582.80                         & 131.54                     & 26888                                      & 491.68                         \\
6000            & 2.50                & 0.54                   & 0.24          & 5     & 249.88                     & 79360                                      & 656.50                         & 143.23                     & 31951                                      & 518.46                         \\
6500            & 2.50                & 0.54                   & 0.24          & 4.80     & 269.98                     & 92430                                      & 728.10                         & 154.43                     & 37126                                      & 571.30                         \\
7000            & 2.49                & 0.54                   & 0.24          & 5.20     & 291.62                     & 108070                                     & 776.50                         & 167.16                     & 43457                                      & 621.85                         \\
7500            & 2.48                & 0.54                   & 0.24          & 5     & 311.65                     & 123240                                     & 831.90                         & 177.91                     & 49027                                      & 641.73                         \\
8000            & 2.48                & 0.54                   & 0.24          & 5     & 333.22                     & 141000                                     & 893.10                         & 190.60                     & 56325                                      & 699.74                         \\
8500            & 2.48                & 0.54                   & 0.24          & 4.80     & 354.44                     & 159240                                     & 942.40                         & 203.12                     & 63857                                      & 724.36                         \\
9000            & 2.47                & 0.54                   & 0.24          & 4.80     & 372.76                     & 176220                                     & 980                         & 213.35                     & 70824                                      & 769.69                         \\
9500            & 2.47                & 0.54                   & 0.23          & 4.80     & 395.30                     & 198090                                     & 997.20                         & 225.65                     & 78760                                      & 792.65                         \\
10000           & 2.46                & 0.54                   & 0.24          & 4.80     & 414.51                     & 217520                                     & 1068.70                        & 237.12                     & 87273                                      & 871.27                        \\
\bottomrule
\end{tabular}
\label{SI_tab:MCN_stat_sz}
\end{table}

We simulate MCNs of different sizes to gain insights of the structural properties. We present the summary statistics of the MCNs of various number of nodes in Tab.~\ref{SI_tab:MCN_stat_gz} to \ref{SI_tab:MCN_stat_sz}. We are interested in the following representative network metrics and these metrics are the average of 10 random realizations of MCNs:
\begin{enumerate}
\item Average path length measures the mean shortest path length among all pair of nodes in MCN.
\item Average clustering coefficient is calculated following the definition in~\cite{watts1998collective} and measures the average cliquishness of individual travelers.
\item Assortativity measures the proclivity of a node to attach to another node of similar degree. This metric quantifies the similarity of the nodes that get into contact.
\item Diameter measures the longest shortest path of the MCNs.
\item $<k>$ is the average unweighted degree of the MCNs.
\item $<k^2>$ is the average second moment of the unweighted degree of the MCNs.
\item $<k_{max}>$ is the maximum degree of the unweighted MCNs.
\item $<d>$ is the average weighted degree of the MCNs.
\item $<d^2>$ is the average second moment of the weighted degree of the MCNs.
\item $<d_{max}>$ is the maximum weighted degree of the MCNs.
\end{enumerate}

\begin{figure}[!ht]
	\centering
	\includegraphics[width=0.3\linewidth]{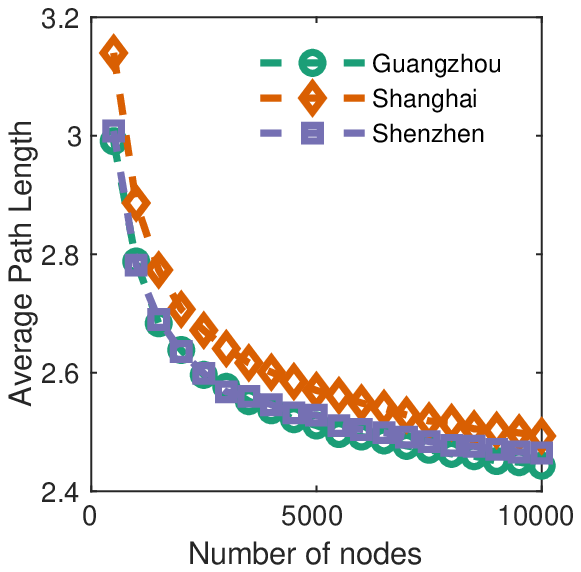}
	\caption{Change of average path length with increasing number of nodes in MCNs.}
	\label{SI_fig:converge_path}
\end{figure}

\begin{figure}[!ht]
	\centering
	\includegraphics[width=0.4\linewidth]{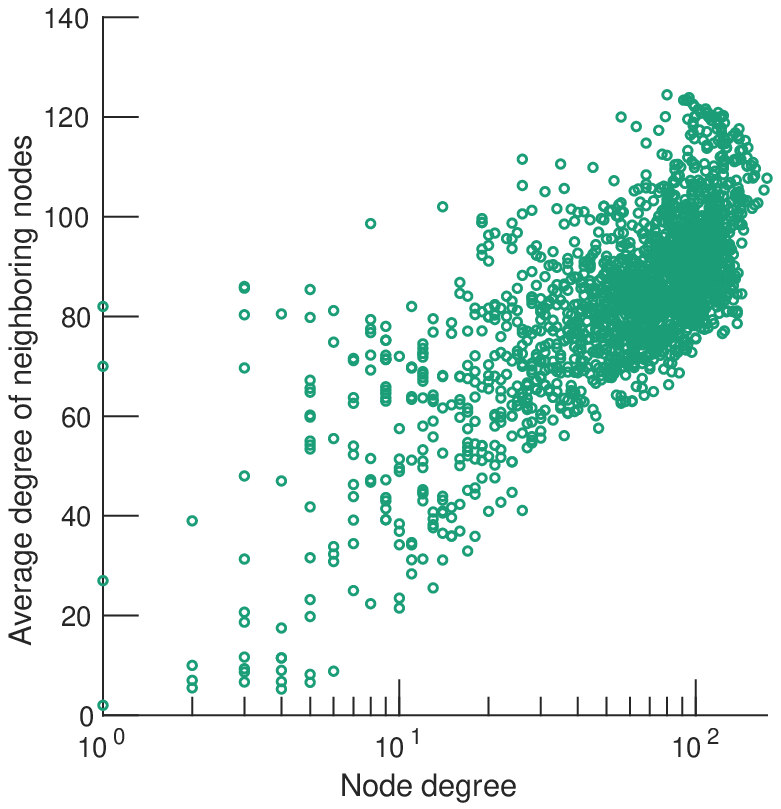}
	\caption{Correlation between node degree and the average degree of neighboring nodes. The results are obtained from a sample MCN with 1000 nodes for Guangzhou during 8:00-8:30 AM.}
	\label{SI_fig:gz_deg_corr}
\end{figure}

Despite the differences in scale and layout of the metro networks, we can immediately observe several structure properties that are universal across the MCNs. The MCNs of different cities and of various number of nodes all present high values of average clustering coefficient, short average path lengths and small network diameters. Moreover, these statistics are found to converge to fixed values with the number of nodes increases from 500 to 7,000 and then become invariant with further increases in the number of nodes in the network (see Fig.~\ref{SI_fig:converge_path} for the convergence of average path length). These results suggest that the structural properties of the MCNs are primarily determined by the layout and scale of the metro network. And the minor differences in the values of these network metrics are also reflections of the differences in their metro systems. Since Shanghai has the largest metro network, we observe the average path length and the network diameter are in general higher than those of Guangzhou and Shenzhen, and the average clustering coefficient is comparatively lower than other cities due to more diverse destinations among travelers. Finally, the assortativity values of the three cities imply that MCNs are weakly assortative where nodes are likely to be connected to other nodes with similar degree and this can be verified from the visualization in Fig.~\ref{SI_fig:gz_deg_corr}. In general we observe a positive correlation between the node degree and the average degree of the neighboring nodes, but there is also huge discrepancy among the average degree of the neighboring nodes for the nodes of similar degree. This indicates certain level of randomness in the number of contacts in the MCN, which is likely to depend on the time of arrival and the specific pair of trip origin and destination.

\begin{figure}[!ht]
	\centering
	\includegraphics[width=\linewidth]{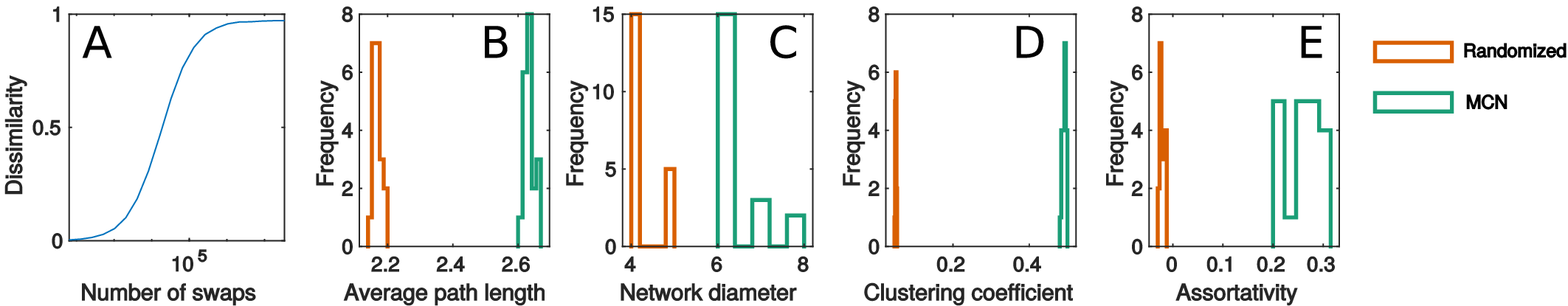}
	\caption{Randomization of simulated MCNs with 1000 nodes using data from Guangzhou during 8:00 to 8:30 AM. (A) presents the dissimilarity between the randomized network and the original MCN with increasing number of swaps. The two networks become almost completely dissimilar with $2^{22}$ swaps. We compare the distribution of the average path length (B), network diameter (C), clustering coefficient (D) and assortativity (E) before and after the randomization using 20 samples of the simulated MCNs.}
	\label{SI_fig:randomize}
\end{figure}

\subsection{Network randomization}
\label{SI:randomization}
To verify the statistical significance of the network metrics for MCNs, we conduct the randomization of the simulated MCNs by selecting two random links in the MCN and swap their endpoints, which is also known as XSwap~\cite{hanhijarvi2009randomization}. XSwap produces the randomization of the MCN while preserving the same degree distribution. We compare the network metrics discussed above between 20 samples of MCN and the corresponding randomized networks, and the results are shown in Fig.~\ref{SI_fig:randomize}. It is obvious that the differences in the distributions of network metrics between the two networks are statistically significant. This confirms the observed structural properties are distinct in MCNs. These results highlight that MCN is a special type of network that presents universal structural properties though being stemmed from metro systems of different scales and layouts.

\subsection{Degree distribution}
The universality of the MCN can also be observed from its unweighted and weighted degree distributions. We first observe that with the same number of nodes in MCNs, the average unweighted degree and weighted degree are different among the three cities. Shenzhen metro has the highest average unweighted and weighted node degree and also the largest variation of node degree, followed by Guangzhou and Shanghai respectively. This is likely because that Shenzhen has the smallest metro network among the three which results in higher chance of contact and hence higher clustering coefficient and average node degree. But the probability density function for the degree distributions of the three cities present striking similarities. As shown in Fig.~\ref{SI_fig:mcn_vs_random_vs_powerlaw}, the unweighted degree distribution shows that there is large proportion of nodes of degree smaller or equal to around half of the maximum degree in MCN and it has a tail that decays almost exponentially fast. One may be tempted to fit a power-law distribution to explain the decay of the tail. Indeed, many real networks are observed to be well explained by power-law distribution and we observe alike decaying trend between MCN and the power-law counterpart. And the MCNs are shown to be significantly different from the random networks in terms of the overall shape and the tail behavior. But there are two subtle differences that prevents the use of power-law distribution for characterizing the degree distribution of MCNs. First, as seen in Fig.~\ref{SI_fig:mcn_vs_random_vs_powerlaw}, the chance of having high-degree hubs in MCN is much lower as compared to the scale-free network of the same number of nodes and links. This indicates the decay of the tail is faster than that in the scale-free network. But more importantly is that MCNs are deemed to be scale-dependent and the degree distribution is closely associated with the number of nodes or equivalently the number of travelers in the metro system. This poses a fundamental contradiction to the philosophy behind the power-law distribution and its properties.

\begin{figure}[!ht]
	\centering
	\includegraphics[width=0.8\linewidth]{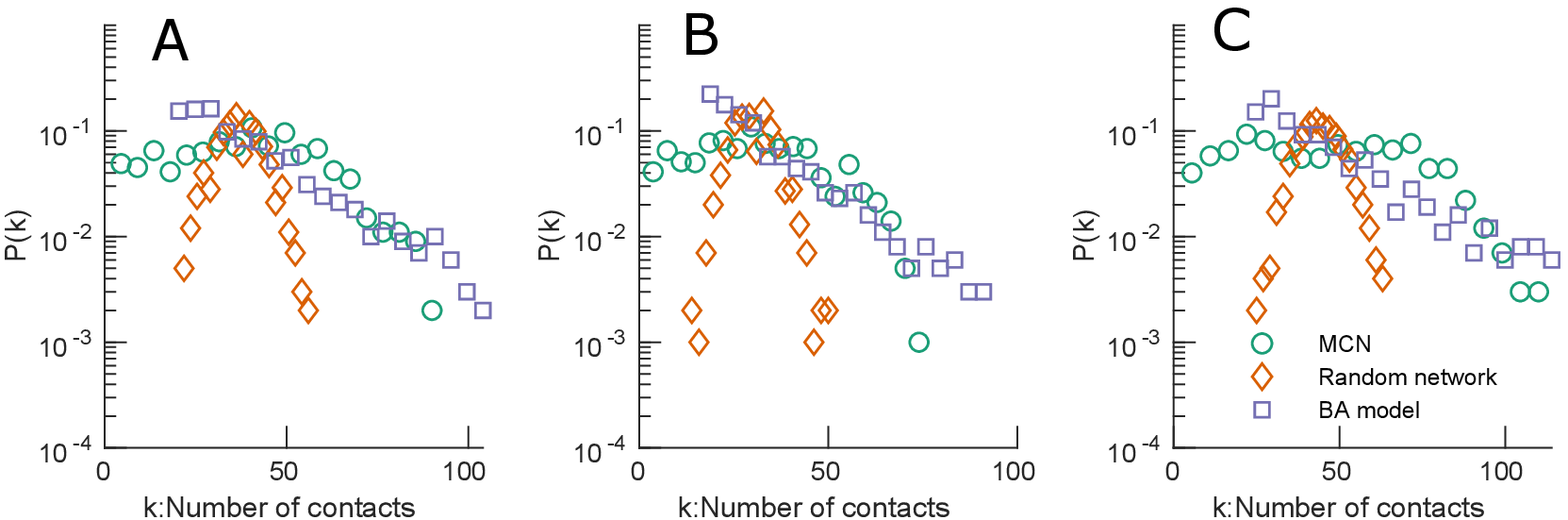}
	\caption{Comparing the probability density function of the degree distribution of MCNs with the random networks and scale-free networks (generated using the Barabasi–Albert (BA) model) with same number of nodes (N=1000) and links for (A) Guangzhou, (B) Shanghai and (C) Shenzhen.}
	\label{SI_fig:mcn_vs_random_vs_powerlaw}
\end{figure}

\begin{figure}[!ht]
	\centering
	\includegraphics[width=0.8\linewidth]{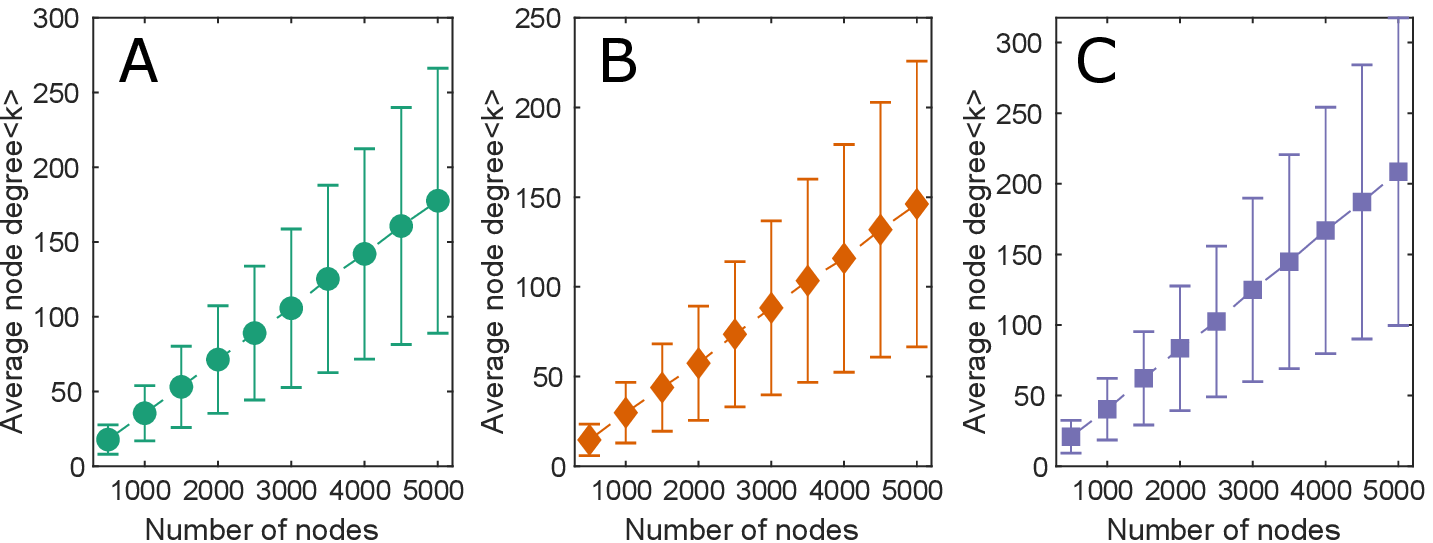}
	\caption{Change of node degree and the standard deviation of node degree with increasing number of nodes in MCNs. (A) Guangzhou, (B) Shanghai and (C) Shenzhen. }
	\label{SI_fig:deg_diverge}
\end{figure}

Nevertheless, the degree distribution of MCN presents several surprising properties that are usually seen in the scale-free network. The first is the possible divergence of $<k^2>$ and $<d^2>$ as shown in Fig.~\ref{SI_fig:deg_diverge}, where the standard deviation of the node degree increases with higher average node degree. Such phenomenon is one important reason that leads to the presence of scale-free property and this is observed in the MCNs for all three cities. In addition, despite seeing that the chance of large hubs is much lower in MCN than in the scale-free network, we empirically observe that the maximum degree of the MCN also increase linearly with increasing size of the network. This again is a unique property that is found in scale-free networks:
\begin{equation}
k_{max}\propto N^{\frac{1}{\gamma-1}}
\end{equation}
where $\gamma_t$ is the exponent  the power-law distribution.

In summary, by analyzing the simulated MCNs, we find that several structural properties of the MCNs are invariant to the size of the network and are primarily determined by the layout and scale of the metro systems. We show that these properties are rare in random networks and are likely to be distinct features of MCNs that are universal across different cities. But more importantly, while presenting fundamental differences when compared with scale-free networks, the MCNs also present universal structural properties that are usually found in scale-free networks. These findings define the MCN as a special class of networks that arises from the collective behavior of travelers and also the interplay between trip pattern and the metro system layout.

\section{Individual level disease transmission model}
Based on the constructed MCNs, we next model the percolation of communicable diseases on the MCNs with the individual based model (IBM). The IBM is adapted from the non-linear dynamical system approach in~\cite{chakrabarti2008epidemic}. In the IBM, each traveler is a node in MCN and the transmission takes place between two travelers with positive $\beta_{ij}$. We consider the classical susceptible-infectious-susceptible (SIS) model as the disease dynamics, while a more refined model such as SIR and SEIR can also be embedded.

The IBM model takes the following items as model input:
\begin{enumerate}
\item The unweighted adjacency matrix $G$ or weighted adjacency matrix $A$ of the MCN.
\item Disease parameters: unit transmission rate $\beta$ and recovery rate $r$, where $1/r$ represents the units number of time steps required for full recovery.
\end{enumerate}

\subsection{Disease transmission rate}
For communicable diseases that spread upon contact, it is well understood that the exposure duration and contact distance between two individuals are two contributing factors to a successful transmission. In our study, the strength of transmission between two individuals is denoted by equation~\ref{SI_eq:strength}. It measures the expected contact duration of two travelers based on their travel profile, and scales the probability of contact by considering the chance if two individuals are within effective transmission distance. As a consequence, we are able to measure the heterogeneous transmission rate between two travelers.
\subsection{The model}
\label{SI:IBM}

Denote $p_{i,t}$ as the probability that node $i$ is infected at time $t$. When an individual $i$ travels, the probability that $i$ stays healthy at time $t$ can be written as:
\begin{equation}
1-p_{i,t}=(1-p_{i,t-1})q_{i,t}+p_{i,t-1}r
\label{SI_eq:not_infect}
\end{equation}
where $q_{i,t}$ represents the probability that neighboring nodes of $i$ fail to transmit disease to node $i$ at time $t$. The first term on right hand side of the equation implies the node was healthy at time $t-1$ and is not infected at time $t$, and the second term suggests that the node was infected at time $t-1$ but recovered at time $t$.

The probability that all neighbors of $i$ failed to transmit the disease can be written as:
\begin{equation}
q_{i,t}=\prod_{j\in\mathcal{N}(i)} (1-p_{j,t}+(1-\beta_{i,j})p_{j,t})
\end{equation}
with $\mathcal{N}(i)$ denotes the set of neighboring nodes of $i$. The right hand side also also contains two parts: either a neighbor $j$ is not infectious at current time $t$ ($1-p_{j,t}$), or if $j$ is infectious but fails to transmit the diseases.

By rearranging equation~\ref{SI_eq:not_infect}, we can express the probability that node $i$ is infected at time $t$ as
\begin{equation}
p_{i,t}=1+p_{i,t-1}(q_{i,t}-r)-q_{i,t},\forall\,i\in V
\end{equation}
And the entire system dynamics over the MCN can be expressed in the matrix form as
\begin{equation}
\textbf{P}_t=\mathcal{G}(\textbf{P}_{t-1})
\label{SI_eq:disease_system}
\end{equation}
So that the disease spreading on MCN is characterized as a non-linear dynamic system.
\subsection{Condition for disease free equilibrium}
The disease dynamic system on MCN has two equilibrium states. One is the disease free equilibrium (DFE), where each individual is in S (healthy) state and the disease is completely eliminated. On the contrary is the endemic equilibrium, where there will always be a positive portion of nodes that are in I state. Formally, the DFE can be defined as
\begin{defn}[Disease free equilibrium (DFE)]
The system reaches the disease free equilibrium if at time $t$ $p_{i,t}=0$ for all nodes.
\end{defn}

The vulnerability of a metro system therefore corresponds to the stability condition for the IBM of the MCN to reach DFE. The stability of the system depends on how the system may return to equilibrium under small perturbation. If the perturbation diminishes and the system goes back to the equilibrium point, the DFE point is said to be asymptotically stable, otherwise the system will reach the endemic state. Before we establish the stability condition for DFE, we first introduce the Gershgorin Circle theorem~\cite{golub2012matrix} as follows

\begin{theorem}[Gershgorin circle theorem]
Every eigenvalue of a complete matrix A lies within at least one of the Gershgorin discs $D(a_{i,i},R_{i})$:
\begin{equation}
{\displaystyle |\lambda -a_{i,i}|=\left|\sum _{j\neq i}a_{i,j}x_{j}\right|\leq \sum _{j\neq i}|a_{i,j}||x_{j}|\leq \sum _{j\neq i}|a_{i,j}|=R_{i}.}
\end{equation}
where $\lambda$ is the eigenvalue of $A$.
\label{theorem:disc}
\end{theorem}

Based on Gershgorin circle theorem, we develop the following condition for the stability of the DFE on MCN:
\begin{prop}
The DFE is asymptotically stable if $max_i(\sum_j(\beta_{i,j}))< r$.
\label{prop:DFE}
\end{prop}
\begin{proof}
We proof the proposition by linearizing the non-linear dynamic system $\mathcal{G}(P_{t-1}=0)$ at the DFE and measuring the partial derivatives $K$:
\begin{equation}
K=\frac{\partial \textbf{P}_t(0)}{\partial \textbf{p}_{t-1}}
\end{equation}
where we have
\begin{equation}
K_{i,j}=-r+1,\quad \text{if $i=j$}
\end{equation}
\begin{equation}
K_{i,j}=\beta_{i,j}, \quad \text{if i$\neq$j and i,j are adjacent}
\end{equation}
\begin{equation}
K_{i,j}=0, \quad \text{o.w.}
\end{equation}
Therefore we have
\begin{equation}
K=(1-r)I+B
\end{equation}
For the DFE to be stable, it must be satisfied that the largest eigenvalue of $K$ is less than 1:
\begin{equation}
\rho(K)< 1
\label{SI_eq:egv}
\end{equation}
Define $\delta$ as an upper bound of the eigenvalue of $K$. Since all diagonal entries of $K$ are identical, by applying Theorem~\ref{theorem:disc}, we have
\begin{equation}
\rho(K)\leq \delta= \text{max}_i(R_i(B))+K_{i,i}= \text{max}_i(\sum_j(\beta_{i,j}))+1-r
\end{equation}
To satisfy the condition in equation~\ref{SI_eq:egv}, we require the upper bound $\delta$ to satisfy:
\begin{equation}
\text{max}_i(\sum_j(\beta_{i,j}))+1-r<1
\end{equation}
This gives that $max_i(\sum_j(\beta_{i,j}))< r$ and completes the proof.
\end{proof}

Proposition~\ref{prop:DFE} has several important implications. The risk level of the MCN is shown to be dictated by the individual who has the highest risk exposure. As long as the exposure rate of this particular individual is smaller than the recovery speed, the system will reach DFE. Otherwise the system may be either DFE or endemic. However, in practice, if we would like to control the spread of communicable diseases, it is unlikely that we may identify who exactly this person is. Even if this person is spotted, vaccine/quarantine the individual does not necessarily reduce the risk level of the overall system, since the second riskiest person may have similar level of risk exposure. This implies that we would also need to examine the structure of the contact network to devise feasible control strategies. In addition, the model provides the solution to monitor the vulnerability of metro systems at very fine scale and identify the periods of time that are of particularly high risk level.

\section{OD-level model}
\label{SI:ODMF}
One drawback of the IBM model, however, comes from its computational bottleneck. It will be very expensive to generate large-scale MCN with millions of passengers that copes with the passenger demand level in real-world scenarios. In this regard, we also develop a metapopulation model based on the flow of travelers between each pair of origin and destination (OD) pair. The OD level model can be used to monitor the risk level of metro systems, but it does not reveal any insights on the contact pattern among individual travelers. The OD-level model treats each pair of OD as the set of nodes and the contagion pattern between OD pairs as the set of links. It can be readily seen that the total number of OD pairs in a given metro network is the square of number of stations, which is much more scalable as compared to constructing contact networks for millions of travelers. Denote $S_i$ and $I_i$ as the susceptible population and infected population of OD pair $i$, and let $\mathcal{P}$ be the set of OD pairs in the network, we have the following equations
\begin{equation}
E_{i,j}=\beta\bar{d}_{j,i}S_iI_j
\end{equation}
where $E_{i,j}$ represents the proportion of susceptible population of $i$ being infected by the infectious population of $j$ and $\bar{d}_{j,i}$ is the expected contact duration between OD pairs $i$ and $j$. The disease dynamics at the OD level can therefore be written as:
\begin{equation}
\frac{d I_i}{d t} = -r I_i + \sum_{j\in\mathcal{P}} E_{i,j}, \forall i\in\mathcal{P}
\label{SI_eq:dynamic1}
\end{equation}
Moreover, since $I_i+S_i=N_i$, where $N_i$ is the total number of travelers for OD pair $i$, equation~\ref{SI_eq:dynamic1} can be further rewritten as:
\begin{equation}
\frac{d I_i}{d t} = -r I_i + \sum_{j\in\mathcal{P}} \beta\bar{d}_{j,i}N_iI_j-\sum_{j\in\mathcal{P}} \beta\bar{d}_{j,i}I_iI_j, \forall i\in\mathcal{P}
\label{SI_eq:dynamic2}
\end{equation}
And the matrix form is therefore
\begin{equation}
\frac{d I}{d t} = FI+b(I)
\label{SI_eq:dynamic3}
\end{equation}
where $F$ is a square matrix with its entry: $F_{ii}=\beta\bar{d}_{i,i}N_iI_i-r$ and $F_{ij}=\beta\bar{d}_{j,i}N_iI_j$. $b(I)$ is a column vector with its entry being $b(I)_i=-\sum_{j\in\mathcal{P}} \beta\bar{d}_{j,i}I_iI_j$. Equation~\ref{SI_eq:dynamic3} gives the disease dynamics at the OD level.
%

We can see that one important difference between the OD-level model and the IBM is that the OD-level model use $\bar{d}_{ij}$ as the aggregate representation of the contact duration between all travelers of OD pair $i$ and travelers of OD pair $j$, rather than the individual level contact duration $d_{ij}$ between travelers $i$ and $j$. As a result, it sacrifices the fidelity for modeling disease at the individual model, but can be used for understanding the system level dynamics more efficiently.

\section{Generation model}
\label{SI_section:generation_model}
By observing that multiple metro networks in different cities share very similar degree distributions in their MCNs,
we next establish the generation mechanism to model how the MCNs are shaped during travel. The goal is to build a single generation mechanism that is capable of restoring the MCNs of all cities to support the universality of MCNs.

Being different from many other networks, the MCNs are special in the way that \textbf{the number of links each node being adjacent should be a function of the total number of nodes in the MCNs.} This corresponds to the congestion effect in the metro system with more number of travelers. Consequently, the MCNs may not be generated in the way like preferential attachment~\cite{newman2001clustering} where new nodes and links are added sequentially. Instead, we follow a process where we first estimate the total number of links in the network and then assign the links among the nodes in a way similar to the configuration model.

To cope with the congestion effect in MCNs, we consider the expected number of contacts each node may encounter as:
\begin{equation}
c_i=  \alpha t_i^{\gamma_t} (N-1)
\label{SI_eq:individual_contact}
\end{equation}
This states that the number of contacts is proportional to the rescaled travel time $t_i^{\gamma_t}$ and the number of nodes $N$. In particular, with $0<\gamma_t\leq 1$, $t_i^{\gamma_t}$ suggests more expected number of contacts with increasing travel time, and if we take the derivative of $t_i^{\gamma_t}$ with respect to $t_i$, we have
\begin{equation}
\frac{d t_i^{\gamma}}{dt_i}=\gamma_t t_i^{\gamma_t-1}
\end{equation}
which implies that, in contrast to the rescale as $\gamma t_i$, the number of contacts does not increase linearly with increasing $t_i$. Instead, the increase rate will drop with the increase in travel time. Indeed, the number of contacts one may have with 40 minutes of travel should not be 10 times that of the number contacts as if one travels for 4 minutes. Regarding the value of $\gamma_t$, we define it as the similarity coefficient that measures the 'similarity' of travels among all travelers. Higher value of $\gamma_t$ indicates that, on average, a traveler will have higher contact chance with another traveler, e.g., two travelers are more likely to travel in the same direction to the same destination. On the other hand, the value of $\alpha$ reflects the scale of the metro network, with the physical meaning being the contact rate per individual traveler per unit time of travel. As a result, $\alpha$ is a system dependent value and varies across the cities.

With the above definition, we can approximate the total number of links (contacts) in the MCNs as:
\begin{equation}
C=\sum_{i=1}^N c_i
\end{equation}
And if we consider each link as two stubs (half links), denoting $M=2C$ and $m_i=2c_i$, we then assign these stubs to each node based on their contribution $t_i^{\gamma_t}$, where the probability that a randomly chosen stub is adjacent to node $i$ with travel time $t_i$ as:
\begin{equation}
w_i=\frac{m_i}{M}=\frac{t_i^{\gamma_t}}{\sum_{j=1}^N t_j^{\gamma_t} }
\end{equation}
While each stub counts as one degree for each node, we can therefore write down the probability density function that a node of travel time $t_i$ is of degree $k$ follows the binomial distribution:
\begin{equation}
 p(k|t_i)={M \choose k} w_i^k (1-w_i)^{M-k}
\end{equation}
Then for a randomly selected node in the MCN, the probability density function for the number of contacts follows
\begin{equation}
p(k)=\sum_{i=1}^N p(t_i){M \choose k} w_i^k (1-w_i)^{M-k}
\end{equation}
where $p(t_i)$ is the probability density function for the human mobility within metro system and is observed to be well captured by the exponential distribution. With large $M$, we can approximate the binomial distribution as the Poisson distribution and hence we have
\begin{equation}
p(k)=\sum_{i=1}^N \frac{(Mw_i)^ke^{-Mw_i}}{k!}p(t_i)
\label{SI_eq:generation_unweighted}
\end{equation}
and this gives the probability density function for the unweighted MCN.

Following equation~\ref{SI_eq:generation_unweighted}, we can subsequently generate a MCN with $M$ stubs attached to each node. To produce the weighted MCN, we follow the process of the configuration model as:
\begin{enumerate}
	\item Randomly selected two stubs in the unweighted MCN, with the nodes adjacent to the stubs as $i$ and $j$.
	\item Connect the selected stubs with an link, and assign the weight to the link: $d_{ij}\propto min(t_i,t_j)$.
	\item Repeat the above two steps until all stubs are exhausted. Output the weighted MCN.
\end{enumerate}
In summary, the above procedure describes the growth of MCN as a two-stage process where we first determine if two travelers will get into contact and then decide the duration of their contact which is assumed to be proportional to the shorter travel time of the two.

\subsection{Validation}
\label{SI:parameter_validation}
The validation of the generation model involves the calibration of the model parameters and then verify if the calibrated generation model is representative of the simulated MCNs from the smart card data. To correctness of the generation model is validated using the two sample Kolmogorov–Smirnov (KS) test~\cite{massey1951kolmogorov} to compare the CDF of the degree distribution of the MCN from the generation model and the CDF of the degree distribution of the MCNs simulated from the smart card data. The null hypothesis of the KS test is that the two data samples for comparison are drawn from the same continuous distribution. Specifically, we binned the degree distribution of each MCN into 20 equal-length intervals, and conduct KS test on the probability distribution of the binned data.

The calibration of model parameters is to find the best $\gamma_t$ and $\alpha$ value that leads to the best goodness of fit between the generation model and the simulated MCNs. In particular, we have different $\gamma_t$ for different time intervals while $\alpha$ is held the same across all time intervals for a particular city. Since we do not have a closed form representation for the probability density function of the weighted MCNs, we conduct cross-validation to find the pair of parameters that minimizes the KS statistics. For each city, we consider $\alpha$ being time invariant since it captures the impacts of metro network structure,  and $\gamma_t$ will change over time to reflect temporal variations of passenger trip patterns. We perform cross-validation to determine the optimal $\alpha$ and $\gamma_t$ for each city and for each time period, with the selection criteria being the parameter combination that gives the lowest sum of KS statistics for weighted degree distribution and unweighted degree distribution of the MCNs.

\begin{figure}[!ht]
	\centering
	\includegraphics[width=0.7\linewidth]{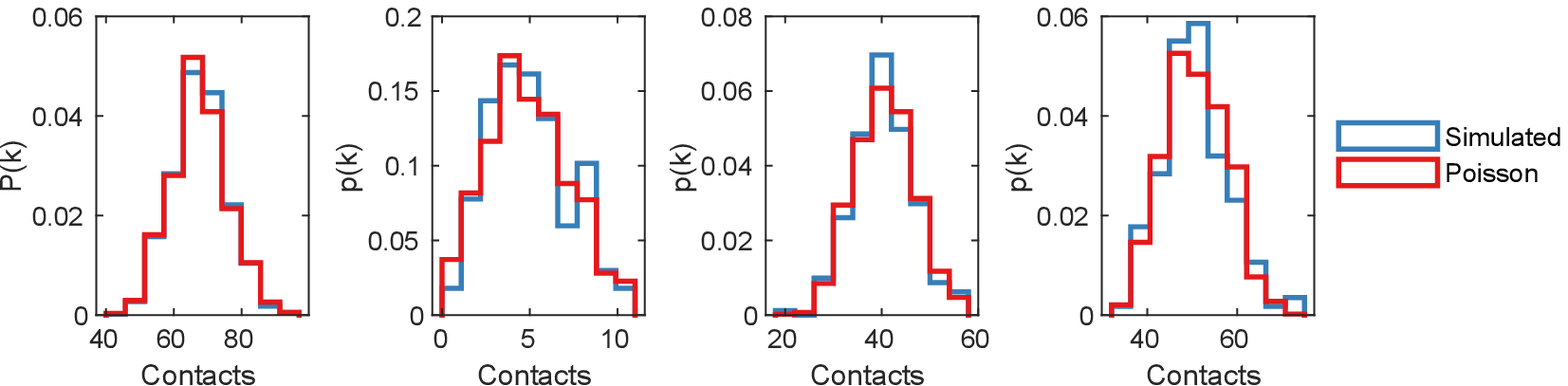}
	\caption{Validation of the Poisson distributed number of contacts.}
\end{figure}

\begin{figure}[!ht]
	\centering
	\includegraphics[width=0.9\linewidth]{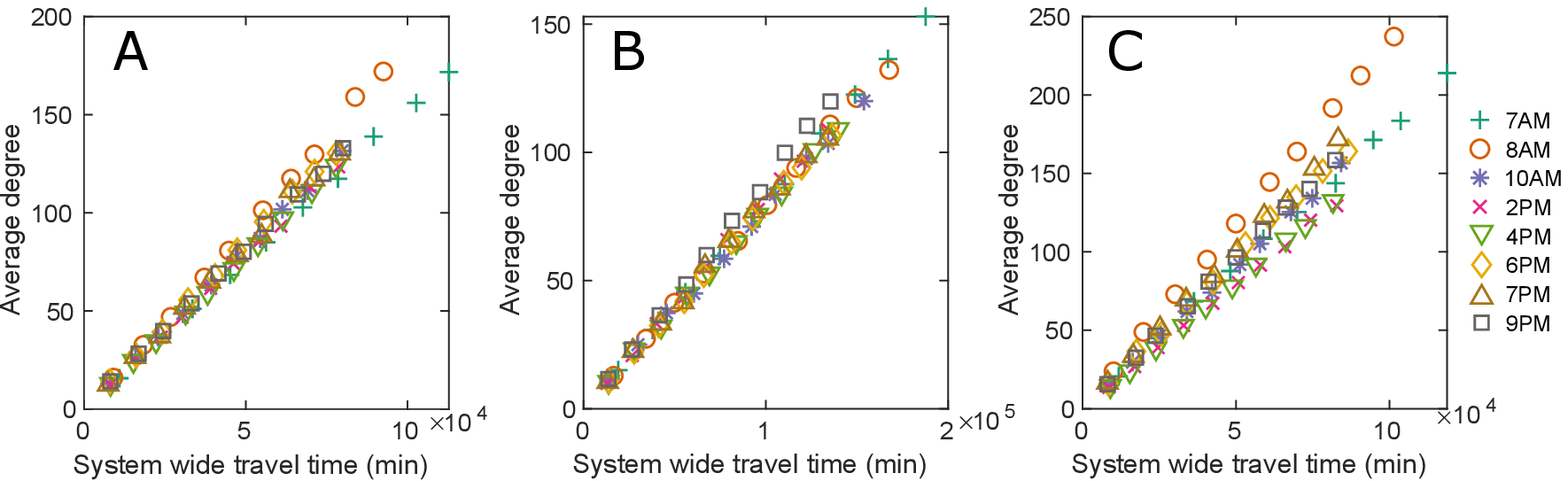}
	\caption{Relationship between the number of contacts per individual and the total system travel time. (A) Guangzhou, (B) Shanghai and (C) Shenzhen. The change of travel time reflects the increase in number of nodes from 500 to 5000 with an increment of 500 at each step. }
	\label{SI_fig: contacts_travel_time}
\end{figure}

\subsection{Measure of similarity}
\label{SI:similarity_measures}
To further validate the correctness of $\gamma_t$, we present the metric for measuring the trip similarity among all travelers and we compare the computed metrics to the fitted $\gamma_t$ values. The similarity is measured by first constructing the correlation matrix $Q$ of the trip pairs, where for each entry of $Q$:
\begin{equation}
Q_{i,j}=f_{i}f_{j}\sigma_{i,j}
\end{equation}
where $f_{i}$ and $f_{j}$ represent the normalized trip flow traveling on OD pair $i$ and OD pair $j$ (number of nodes on OD pair $i$ divided by total number of nodes in MCN). $\sigma_{i,j}=\frac{d_{i,j}}{max d_{i,j}}$ refers to the standardized contact duration between the two OD pairs. In this regard, each entry of $Q$ measures the pairwise contagion strength and row $i$ of $Q$ therefore gives the level of correlation of OD pair $i$ with all other OD pairs. And the correlation depends on both the demand level as well as the contact duration.

Given the correlation matrix $Q$, we next extract the top $n$ eigenvalues of $Q$ as ${\lambda_1, \lambda_2,...,\lambda_n}$, and we define similarity index as the standard deviation among the top $n$ eigenvalues:
\begin{equation}
s=\sqrt{\frac{\sum(\lambda_i-\bar{\lambda})^2}{n}}
\end{equation}
where $\bar{\lambda}=\sum \lambda_i/n$ refers to the mean of the eigenvalues. With the normalization of trip flow and standardization of contact duration, we restrict $s$ to lie between 0 and 1.

The idea of similarity index is related to the principal component analysis of the correlation matrix, where the trace of the correlation matrix measures the total variance and the top $n$ principal components seek to maximize the variance. The standard deviation among the top $n$ eigenvalues therefore measures the differences in the total contribution to the total variance of each component, and hence reflects the trip differences among passengers. In particular, if those trips are totally uncorrelated and the trips, then the eigenvalues are all of the same value and the standard deviation among them is simply 0. An example in metro system is that the demand are evenly distributed among all OD pairs and these OD pairs have no overlapping segments to enable contacts. On the other hand, if some of the trip pairs are highly correlated, we should have few eigenvalues of value much higher than others, which gives rise to the large standard deviation. An extreme case in the metro system is that all travelers leave from the same origin to the same destination so that these trips are perfectly correlated and the standard deviation is therefore 1. The correlation matrix for metro systems is of size $N^2 \times N^2$ ($N$ here is the number of stations) and we do not need to compute all $N^2$ eigenvalues. Instead, based on empirically observations, we find that eigenvalues drop quickly to nearly zero and we therefore set $n=200$.

\section{First and second moment of node degree in MCN}
Here we develop the first and second moment of the MCNs based on the generation model. We also derive an approximation of the largest node degree in a given MCN. These help to gain further insights on the degree distribution of the MCNs with increasing number of nodes. 
\label{SI:first_second_moment}
\subsection{First moment $<k>$}
From the generation model for MCN we have:
\begin{equation}
p(k)= \sum_{i=1}^N e^{-Mw_i}\frac{(Mw_i)^k}{k!}p(w_i)
\end{equation}

And we can calculate the average degree of the MCN as:
\begin{equation}
<k>=\sum_{k=1}^K p(k)k=\sum_{k=1}^K k\sum_{i=1}^N e^{-Mw_i}\frac{(Mw_i)^k}{k!}p(w_i) = \sum_{i=1}^N p(w_i) \{\sum_{k=1}^K k e^{-Mw_i}\frac{(Mw_i)^k}{k!} \}
\end{equation}
When $K\rightarrow\infty$, the summation of discrete degree can be replaced with the integration:
\begin{equation}
\sum_{k=1}^K ke^{-Mw_i}\frac{(Mw_i)^k}{k!}=\int_{0}^{K} ke^{-Mw_i}\frac{(Mw_i)^k}{k!} dk = Mw_i
\end{equation}
where $M_wi$ is the mean of the binomial distribution of $M$ trials and $w_i$ rate of success. And we therefore have:
\begin{equation}
<k>=\sum_{i=1}^N p(w_i) Mw_i
\end{equation}
Note that $p(w_i)=p(t_i)$ represents the probability density function for human mobility in metro network, which we find to be approximated by an exponentially decaying tail. We consider that
\begin{equation}
Mw_i=\alpha (N-1) t_i^{\gamma_t}, \quad p(t_i)=be^{-t_i/\lambda}
\end{equation}
Then
\begin{equation}
<k>=\sum_{i=1}^N \alpha (N-1) t_i^{\gamma_t} be^{-t_i/\lambda} \approx \int_0^{t_{max}} \alpha (N-1) t^{\gamma_t} be^{-t/\lambda} dt
\end{equation}
Where the integration gives
\begin{equation}
\int t^{\gamma_t} be^{-t/\lambda} dt =-b\lambda^{\gamma_t+1}\Gamma(\gamma_t+1,\frac{t}{\lambda})+ C
\end{equation}
where $\Gamma(m,n)$ is the upper incomplete Gamma function with $\Gamma(m,n)\rightarrow 0 $ if $n\rightarrow \infty$, and $\Gamma(m,0)=\Gamma(m)$. This suggests that
\begin{equation}
<k> =\alpha (N-1) b \lambda^{\gamma_t+1} \Gamma(\gamma_t+1)
\end{equation}
which suggests that the \textbf{average degree of MCN is linearly proportional to the number of nodes in the network}.
\subsection{Second moment $<k^2>$}
In addition, we can also calculate $<k^2>$ as
\begin{equation}
<k^2>=\sum_{k=1}^K p(k)k^2=\sum_{i=1}^N p(w_i)\{Mw_i(1-w_i+Mw_i)\}
\end{equation}
where $Mw_i(1-w_i+Mw_i)$ represents the second moment of the binomial distribution.
Following the same procedure for deriving $<k>$, we arrive at the expression of $<k^2>$ as
\begin{equation}
<k^2>=\alpha^2 (N-1)^2 b \lambda^{2\gamma_t+1} \Gamma(2\gamma_t+1) + O(N)
\end{equation}
\textbf{This implies that the variance of MCN scales quadratically to the increase in number of nodes. These results explain the divergence of $<k^2>$ with $N\rightarrow \infty$}.

\subsection{Max degree node $k_{max}$}
To estimate the maximum degree in MCN, let we consider the probability that
\begin{equation}
\begin{aligned}
\int_{k_{max}}^{\infty} p(k)dk=1-p(k_{max})&= 1-\sum_{i=1}^Np(w_i) \sum_{k=1}^{k_{max}}e^{-Mw_i}\frac{(Mw_i)^{k_{max}}}{k_{max}!}\\
&\approx \sum_{i=1}^Np(w_i)e^{-Mw_i}\frac{(Mw_i)^{k_{max}+1}}{(k_{max}+1)!}\\
&\approx \frac{bB^{k_{max}+1}\Gamma(k_{max}+2)}{(k_{max}+1)!(B+1/\lambda)^{k_{max}+2}}\\
\end{aligned}
\end{equation}
where $B=\alpha (N-1)$.

And for the maximum degree, we expect that
\begin{equation}
\int_{k_{max}}^{\infty} p(k)dk=\frac{1}{N}
\end{equation}
so that there is one node that is within the range $[k_{max},\infty]$. This condition suggests that
\begin{equation}
\frac{B^{k_{max}+1}}{(B+1/\lambda)^{k_{max}+2}}=\frac{1}{Nb}
\end{equation}
By taking the natural log on both sides, the equation simplifies to
\begin{equation}
\begin{aligned}
k_{max}+1&=\frac{ln(Nb)-ln(B+1/\lambda)}{ln(B+1/\lambda)-ln(B)}\\
&\propto \frac{1}{ln(1+1/(\alpha\lambda (N-1)))}\\
&\propto \alpha \lambda (N-1)
\end{aligned}
\end{equation}
where for the last step we make use of the Taylor series expansion for log values
\begin{equation}
ln(1+1/(\alpha\lambda (N-1))=\frac{1}{\alpha\lambda (N-1)}+O(1)
\end{equation}
This result indicates that the \textbf{maximum degree of the MCN is linearly proportional to the number of nodes in the network}.

\bibliographystyle{unsrt}
\bibliography{sample}

%% file: introduction.tex
\section*{Introduction}
The rapid growth of population and activity intensity in megacities have propelled an evolutional shift of urban mobility from individual-centric travels to sustainable urban mobility. This substantial shift concerns environment, energy consumption, equity, among others~\cite{manaugh2015integrating,bertolini2003urban,black2002sustainable,goldman2006sustainable}.
Lying at the heart for promoting sustainable travel is our understanding of the interplay among urban form, transportation system, and human mobility. Research across a diverse stream of studies and data sources have shown the scaling properties of individual mobility~\cite{gonzalez2008understanding,song2010modelling,hasan2013understanding,sun2013understanding}: the vast majority of people travel between a few popular locations and their travel distances are bounded by the scale of the city and its transportation systems.

One indispensable component of urban transportation is the mass transit system, which is so far the only sustainable solution to serve urban mobility needs on a large scale.
In 2018, public mass transit served over 53 billion passengers worldwide. The three busiest metro systems reached the daily ridership of 9.48 million (Tokyo), 6.49 million (Moscow) and 5.6 million (Shanghai)~\cite{metro_report}, respectively. While these systems allow a large number of commuters to travel efficiently, they also result in high population density within close proximity for long durations. These features establish an environment conductive to the spread of communicable diseases~\cite{xie2007far,yang2009transmissibility,salathe2010high}. In particular, pathogens of infectious travelers can migrate to adjacent travelers through droplets and airborne transmissions, resulting in secondary infections during travel. Public mass transit and the underlying travel patterns are becoming an essential catalyst for influenza pandemics and may greatly accelerate the spreading pace of communicable diseases and thus increase the intensity of disease outbreaks in megacities. Despite acknowledging the linkage between human mobility and the spread of infectious disease, current models do not understand the nexus on how physical contacts/encounters among individuals---which in turn enable the transmission of communicable diseases---are created from human mobility.

An attractive approach to address the challenge is to construct the contact networks during travel and then embed the disease percolation process among individual travelers into the contact networks.
Recent advances in complex network theories and epidemic modeling have established a striking connection between network structure and disease dynamics~\cite{pastor2001epidemic1,pastor2001epidemic2,pastor2002immunization,pastor2001epidemic2,meyers2005network,keeling2005networks,meyers2007contact} 
and the epidemic spreading was modeled on various mobility scales including airline network~\cite{colizza2007predictability,colizza2006role}, ground transportation network~\cite{balcan2009multiscale} and river network~\cite{mari2011modelling,gatto2012generalized}. While mobility networks represent mediate channels for epidemic outbreaks, other studies also focused on correlating
disease dynamics with direct human interactions as contact networks at activity locations ~\cite{meyers2005network,salathe2010high,stehle2011simulation,litvinova2019reactive, eubank2004modelling}.~
These studies connected the disease percolation with either mobility networks or contact networks, but the linkage on how contact networks are shaped as a function of human mobility is still missing. Meanwhile, such an interplay has significant implications on keeping infectious people from engaging in daily activities and consequently stopping the epidemic spreading from the source.

%


\begin{figure*}[!h]
    \centering
    \includegraphics[width=\linewidth]{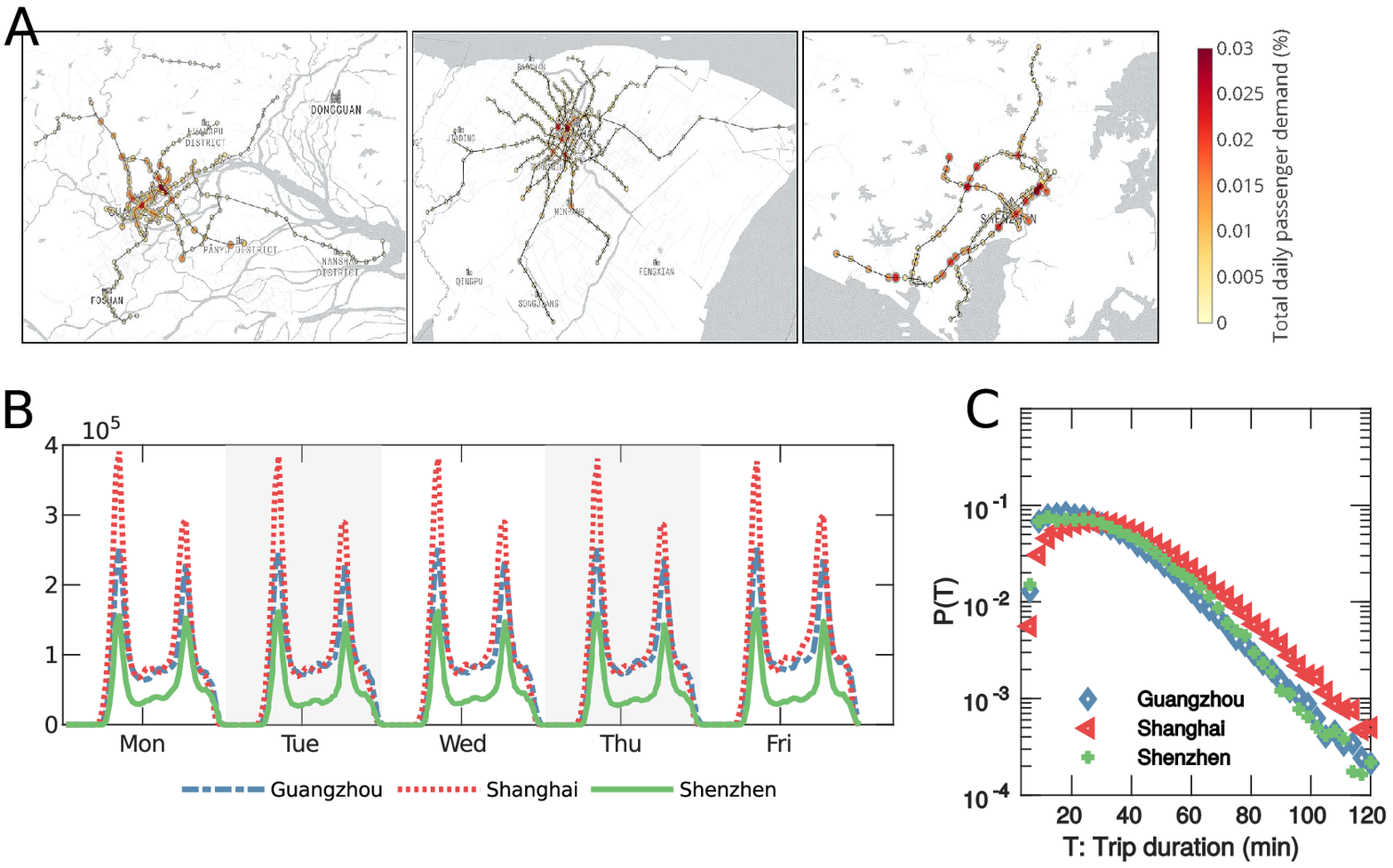}
    \caption{Summary characteristics of the travel pattern within the metro networks. (A) Spatial layout of the metro networks of the three cities (from left to right: Guangzhou, Shanghai, Shenzhen). (B) Temporal distribution of metro ridership obtained from the smart card transaction data. (C) Probability density of function of metro trip duration T. We observe that the distributions metro trip duration of all three cities follow exponentially decaying tails ($p(T)\propto e^{-T/\lambda}$). For Guangzhou, we find that trips with $T>50$min can be well fitted with $\lambda=13.78$min. For Shanghai, trips with $T>50$min can be well fitted with $\lambda=16.59$min. For Shenzhen, trips with $T>50$min can be well fitted with $\lambda=13.43$min.}
    \label{fig:metro_networks}
\end{figure*}

To close this gap, this study aims to characterize how human mobility shapes contact networks during travel and how it subsequently affects the threshold of disease percolation among individual travelers. In previous studies, contact networks were either of high-resolution for small systems (e.g., at conference~\cite{stehle2011simulation} and school~\cite{salathe2010high,litvinova2019reactive}) or of low-resolution for large systems by simulating from survey data~\cite{meyers2005network,eubank2004modelling}. Here we construct high-resolution contact networks for city-wide transit systems by leveraging smart card data from three major cities in China: Guangzhou, Shanghai, and Shenzhen (see SI~\sref{SI:data} for a detailed description of the data). We focus on the metro system, the most used public transit mode, to rebuild the contact networks among metro travelers, but the approach is broadly applicable to other transit systems. The three cities are of drastically different scales and distinct metro network layouts (Fig.\ref{fig:metro_networks}A). Specifically, Shanghai has the largest population, metro system, and the highest metro passenger volume (over 4 million daily records), followed by Shenzhen (2.1 million) and Guangzhou (1.6 million). The metro ridership of the three cities presents highly regular and recurrent patterns during weekdays with prominent travel peaks~(Fig.\ref{fig:metro_networks}B), which implies a large number of daily commuters and repeated metro visits. The large-scale trip data, as the result of intensive daily activities in metro systems, allow us to directly probe the representative mobility patterns of metro travelers in these cities (see Fig.~\ref{fig:metro_networks}C). Despite distinct network layouts, size, and trip demand, the mobility patterns of the three cities are observed to be strikingly similar. We find that there is a large number of travelers with travel time under 50 minutes, and the number of travelers decays exponentially with increasing trip length. This finding holds true across all three cities, with Shanghai having a lower decay rate (16.59min) and the decay rates for Guangzhou (13.78min) and Shenzhen (13.43min) being almost identical. This indicates that trip length in metro systems is bounded by the scale of the metro network, which reflects the scale of a city where trip duration may not grow infinitely as in the scale-free networks. The finding also provides strong evidence to support the universality of human mobility within public transit, where the travel time follows the exponential distribution with the decay rate being proportional to the scale of the city. As physical encounters are driven by human mobility, this motivates us to investigate the possible existence of scaling laws for the contact patterns in public mass transit networks, as the results of the universal mobility patterns.

\begin{figure*}[!h]
    \centering
    \includegraphics[width=\linewidth]{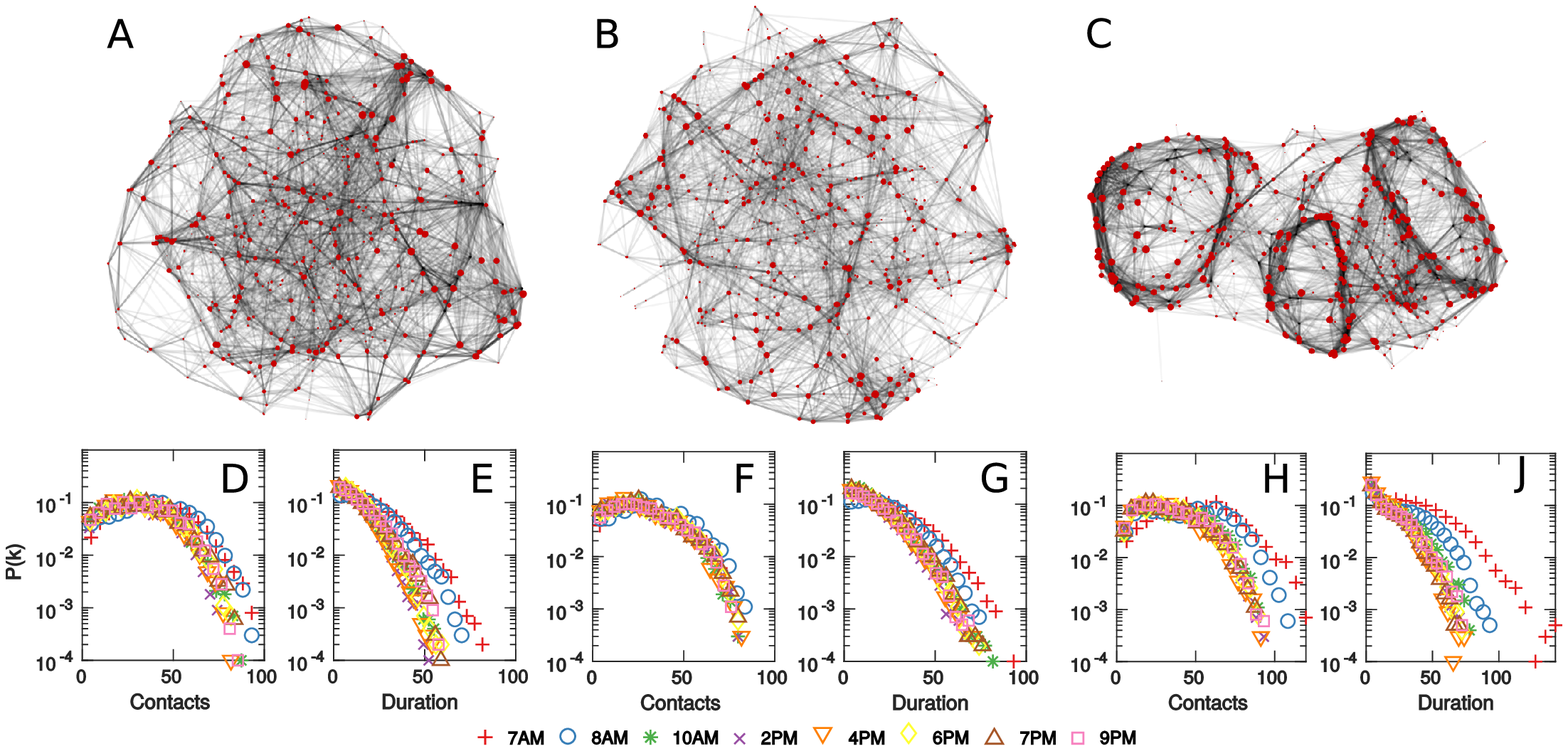}
    \caption{Simulated MCNs with 500 nodes and the unweighted and weighted degree distributions of simulated MCNs with 1000 nodes. (A)-(C) visualizes the layouts of the simulated MCNs in Guangzhou, Shanghai and Shenzhen. In the visualization, larger node size reflects higher node degree and the transparency of the link is proportional to the duration of contact. (D),(F) and (H) present the probability density function of the unweighted degree distributions of Guangzhou, Shanghai and Shenzhen. (E), (G) and (J) present the probability density function of weighted degree distributions of Guangzhou, Shanghai and Shenzhen. }
    \label{fig:comparison_mcn}
\end{figure*}

%% file: metro_contact_network.tex

\section*{Metro Contact Network}
Smart card data only provide the entry and exit information of a trip. To gain insights on how travelers come in contact with others during travel, we develop a simulation model based on the observed metro network layout, demand profile, and mobility patterns. The simulation constructs high fidelity metro contact networks (MCNs) by first sampling passenger arrivals at each metro station and their trip destinations, then calculating if two individuals will come into contact based on their trip profiles, and finally assigning expected contact duration between each pair of individuals (The detailed description of the simulation is presented in the SI~\sref{SI:mcn_generation_algorithm}). The inputs to the simulation are the number of travelers ($N$), the time period of interest, the operational timetable and metro network layout. The simulation then produces a $N\times N$ matrix describing the physical contact pattern between each pair of individual travelers. In particular, each positive entry of the matrix denotes the expected contact duration between two travelers within \textit{effective transmission range}, e.g., two individuals are of close proximity so that the airborne transmission of a communicable disease is likely. For typical droplet transmission, the effective range is less than 3 feet while certain disease such as SARS may reach 6 feet~\cite{siegel20072007}.
\begin{table}[!h]
\centering
\caption{Summary statistics of the generated MCNs with 1000 nodes. The MCNS correspond to the travel pattern during 8-8:30 AM on weekday and the statistics are measured using unweighted MCNs. Results are averaged from 10 random realizations.}
\begin{tabular}{p{1.5cm}p{1cm}p{3cm}p{1.5cm}p{1.5cm}}
	\toprule
          & $<k>$ & Clustering coefficient  & Diameter & CPL  \\
\hline
Guangzhou & 35.09          & 0.49   & 7.2     & 2.79  \\
Shanghai  & 29.14          & 0.46    & 7.6   & 2.89        \\
Shenzhen  & 41.59          & 0.54   & 6     & 2.78     \\
\bottomrule
\end{tabular}
\label{tab:mcn_stat}
\end{table}
We then visualize the structure of the MCNs by simulating a sample realization for each of the cities during 8-8:30 AM, and we set the number of travelers to 500 for better visibility of the network structure (Fig.~\ref{fig:comparison_mcn}A-C). We observe that the MCNs are visually different among the cities, which is due to the differences in metro network layouts. But these MCNs also share several structural commonalities, including local clusters of nodes and the discrepancies of node degree. 

To gain better understanding of their structural properties, for each city, we further generated MCNs from 500 nodes to $10^4$ and the details are summarized in SI~\sref{SI:mcn_generation}. We observe that there is fewer super nodes in the MCNs as opposed to the scale-free networks. Instead, there are a large number of nodes with low to medium degree. And we further confirm the structural similarities by different quantitative metrics, as shown in Tab.~\ref{tab:mcn_stat}. All the MCNs are observed to have high node degree heterogeneity and high clustering coefficient, and have small network diameter and characteristic path length (CPL). These properties are statistically different from the metrics of their random counterparts (see SI~\sref{SI:randomization}), which corroborates the distinct structural properties associated with MCNs. These confirm that MCNs are a type of small-world network~\cite{telesford2011ubiquity}, and this observation leads to significant implications in the context of disease percolation. Specifically, the outbreak of an exceedingly infectious disease may quickly synchronize among all the travelers because of the small-world property, and therefore the metro system becomes highly vulnerable. But unlike many real-world networks, we report that these network metrics are invariant to the size of the MCNs (Tab.~\ref{SI_tab:MCN_stat_gz}-\ref{SI_tab:MCN_stat_sz}), but instead being determined by the metro network layout and human mobility patterns. In this regard, MCNs, like many other real-world contact networks in school, conference sites and major activity locations, should be regarded and studied as the products of the interactions between human mobility and physical infrastructure. 

By plotting the degree distributions of the number of contacts and contact duration for metro travelers (Fig.~\ref{fig:comparison_mcn}D-F), we find an even more remarkable similarity among the MCNs. Despite the differences in metro network layouts, visualizations and statistics of the simulated MCNs, the degree distributions are found to follow a similar distribution across the three cities, and such observation is also valid for different time periods of the day. In particular, the unweighted degree distributions of the MCNs show a large number of nodes of low to medium degrees (e.g., degree smaller than 50) and the node degrees within this range are found to be nearly uniformly-distributed. But with increasing number of contacts and length of contact duration, the tails of the node degree distributions are found to decay exponentially, similar to the observations for metro mobility patterns. In addition, the rates of decay are found to be time-dependent and also differ among the cities. These observations lead to the conjecture of a universal mechanism underlying the contact of metro travelers, and we explore the mechanism in more depth in the following sections.

%% file: disease_dynamics.tex

\section*{Disease dynamics in contact network}
With the reconstructed contact network, the risk of communicable diseases can be quantified by modeling the dynamics of disease percolation among individual travelers. We introduce an individual-based model (IBM) following~\cite{chakrabarti2008epidemic}. To characterize disease dynamics within the contact network, the classical susceptible-infectious-susceptible (SIS) process is embedded in the IBM over MCNs. Unlike previous studies~\cite{pastor2001epidemic1,meyers2005network}, this framework does not require nodes and transmission between nodes to be homogeneous, which allows us to model heterogeneous infectious rates due to the varying contact duration. Denote the probability that node $i$ is infected at time $t$ as $p_{i,t}$ and the recovery rate as $r$, we have
\begin{equation}
p_{i,t}=1+p_{i,t-1}(q_{i,t}-r)-q_{i,t},\forall\,i\in V
\end{equation}
where $q_{i,t}$ represents the probability that node $i$ is in S at time $t$, which depends on that all its neighbors $j\in\mathcal{N}(i)$ are either in S or in I but fail the transmission:
\begin{equation}
q_{i,t}=\prod_{j\in\mathcal{N}(i)} (1-p_{j,t}+(1-\beta_{i,j})p_{j,t})
\end{equation}
In the equation, $\beta_{i,j}=\beta t_{i,j}$ represents the transmission rate between node $i$ and $j$, which takes the product of per unit time disease transmission strength $\beta$ and the contact duration $t_{i,j}$. With $N$ such nodes, we arrive at a nonlinear dynamic system (see SI~\sref{SI:IBM}) with two equilibrium states: (1) the disease-free equilibrium (DFE) where all individuals are in S state (e.g., $p_{i,t}=0$) and (2) the endemic equilibrium where a positive proportion of individuals are in I state. The asymptotic stability of the DFE relies on the network-specific critical threshold $\delta$ that is associated with the largest eigenvalue of the adjacency matrix of MCNs. And we show that the critical threshold is upper bounded by the largest node degree in MCNs as:
\begin{equation}
\delta \leq \max_i \sum_j \beta_{i,j}-r+1=\bar{\delta}
\end{equation}
As a consequence, if $\bar{\delta}< 1$, the disease is guaranteed to go extinct while the disease may be endemic with $\bar{\delta}>1$. Since $\beta$ and $r$ are endogenous parameters, the risk level that pertains to a specific disease primarily depends on the value of $\max_i \sum_j\beta_{i,j}$. Such finding subsequently builds the essential connection between the vulnerability of public mass transit with the degree distribution of its contact networks and identifies the impact of the structural property of MCNs on disease threshold in transit systems. The observation provides two immediate implications. First, we can verify that the risk level of a MCN is driven by the riskiest individual who has the highest number of contacts or contact duration, which concerns the tail pattern of MCNs unweighted and weighted degree distributions. Second, by removing the riskiest individual, the next riskiest person may have a similar risk level according to the revealed degree distributions (Fig.\ref{fig:comparison_mcn}D-J). This highlights the difficulties in improving MCN's vulnerability and stopping the disease during outbreaks.

In addition to the IBM model, we also build an equivalent OD-based mean-field (ODMF) approach that models the disease dynamics on the passenger flow level between each pair of metro stations (SI~\sref{SI:ODMF}). Note IBM is computationally expensive due to the construction of MCNs, and the ODMF can be used to approximately probe the system-wide disease dynamics for the real number of metro travelers.

\section*{Disease control strategies}
The best practice for controlling the disease is to immunize travelers through vaccination and quarantine~\cite{kaplan2002emergency}. We next explore the effectiveness of five immunization strategies with the percentage of individuals immunized as a control parameter. \textit{OD-based} and \textit{station-based} approaches represent population control that immunizes a portion of travelers commuting between a pair of stations or originating from a station. These two control strategies are motivated by the observation that fewer than 20\% of the stations produce over 80\% of travel demand (Fig.~\ref{fig:demand_control}A), and more than 80\% of the metro commuters are associated with fewer than 20\% of station pairs(Fig.~\ref{fig:demand_control}B). These findings suggest that population control may yield satisfactory results in reducing the risk level by focusing on populated stations and trip pairs, as these travelers are likely to have more number of contacts. On the other hand, we also consider \textit{uniform}, \textit{targeted} and \textit{distance based} approaches that are individual-centered methods. The uniform strategy immunizes randomly selected travelers, the targeted strategy iteratively removes travelers of the longest contact duration, and the distance-based strategy aims at immunizing travelers of the longest travel time. Specifically, the targeted method is reported to be most effective in the complex network literature~\cite{pastor2001epidemic1,pastor2002immunization}. The effectiveness of these control strategies is then examined based on the relative risk level (RRL), which measures the reduction in $\max_i \sum_j\beta_{i,j}$ with increasing number of immunized travelers.
\begin{figure}[!h]
    \centering
	\includegraphics[width=0.7\linewidth]{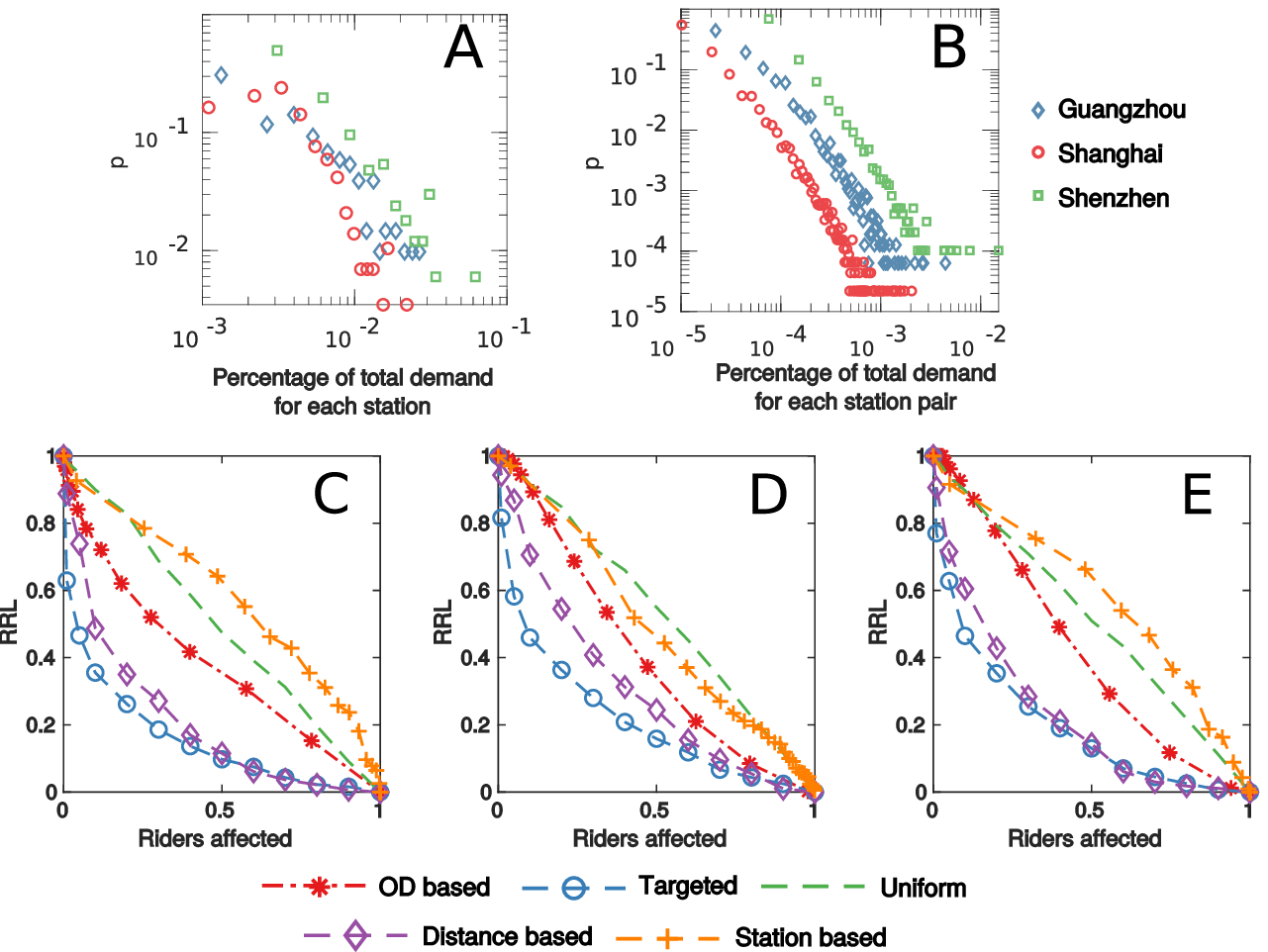}\\
	\caption{Distribution of travel demand and the effectiveness of different control strategies for the MCNs of the three cities. (A) presents the probability density function of the trip demand at the station level. (B) presents the probability density function of the trip demand of each pair of stations. (C), (D) and (E) visualize the effectiveness of OD based, targeted, uniform, distance based and station based control strategies for each city. The effectiveness of control strategies is compared with the proportion of trip demand affected by the corresponding strategy.}
    \label{fig:demand_control}
\end{figure}

Our results suggest that the most effective method is the targeted immunization followed by the distance-based method and the OD based method, and all the three are superior to the uniform immunization. This finding is consistent across the three cities~(Fig.\ref{fig:demand_control}C-E). For the targeted immunization, we observe a 27\% reduction in RRL by immunizing the top 1\% riskiest individuals, and a 60\% reduction in RRL can be achieved by removing the top 10\% riskiest travelers. On the other hand, the station-based method is found to be less effective than the uniform policy in two of the three cities. The primary reason is that populated stations may not be origins of travelers with high number of contacts and contact duration.
Unfortunately, it is usually impracticable to identify the risk level of a traveler before the trip, which is the major barrier for the implementation of targeted immunization. As effective alternatives, we may introduce the distance-based and OD-based strategies by tracking the historical trips of an individual from her smart card. The most practical control strategies are the station-based and the uniform strategies, but neither of them is shown to be effective enough. These results highlight the challenges for stopping the spread of the disease in transit systems. It is therefore important to devise better strategies for the operation of metro systems and for improving the structure of the metro networks.

%% file: generation_model.tex

\section*{A generation model for MCNs}
Here we develop and validate a simple generation model for MCNs. The MCN to be generated is a scale-dependent network where the degree distribution is a function of the total number of travelers in the network. We observe that the travelers' mobility patterns follow the exponential distribution while the contact degree distributions also decay exponentially. Following our discussion on the MCN structure, we hypothesize that (1) contacts are driven by travelers' mobility pattern, and (2) the probability of two travelers getting into contact is bounded by their mobility in the metro network and also the layout of the metro network. We focus on investigating a universal generation model with individuals' mobility pattern as the input and we consider that the number of contacts is proportional to the total travel time $t_i$ of each traveler. Recall that the degree distribution is found to vary across time and city. We thereafter introduce two variables: $\alpha$ captures the impact from metro network layout and $\gamma_t$ models the temporal characteristics of travelers' mobility. We consider the expected total number of contacts experienced by $N$ travelers as:
\begin{equation}
C= \sum_i^N \alpha t_i^{\gamma_t} (N-1)
\label{eq:individual_contact}
\end{equation}
Equation~\ref{eq:individual_contact} accounts for the scale-dependent nature by including $N$ on the right-hand side and $\alpha t_i^{\gamma_t}$ determines the rate that a commuter of travel time $t_i$ will meet other $N-1$ travelers in the system. And $t_i^{\gamma_t}$ refers to the rescaled travel time which depends on the temporal trip similarity among travelers. This is motivated by the fact that different time of day will result in riders heading to various destinations and $\gamma_t$ therefore measures how similar their destinations are. Consider $M=2C$ as the total number of stubs (half-edges) in MCNs, we can derive the probability that a node is of degree $k$ as (see SI~\sref{SI_section:generation_model} for derivation details and Fig.~\ref{SI_fig: contacts_travel_time} on empirical evidence that supports $M$):
\begin{equation}
    p(k)=\sum_{i=1}^N \frac{(Mw_i)^ke^{-Mw_i}}{k!}p(t_i)
    \label{eq:generation_model}
\end{equation}
with $w_i = 2\alpha t_i^{\gamma_t} (N-1)/M$ being the probability that a randomly selected stub is attached to node $i$.
Given the pdf in equation~\ref{eq:generation_model}, the MCNs can then be generated following the configuration model by first sampling the degree sequence $K=\{k_1,k_2,...,k_N\}$ from $p(k)$, then randomly selecting and connecting a pair of stubs until all stubs are exhausted. For each pair of matched stubs between node $i$ and $j$, we further assign the weight $d_{ij}\propto \min(t_i^{\gamma_t},t_j^{\gamma_t})$ as the edge weight and obtain the weighted degree distribution.

To calibrate $\alpha$ and $\gamma_t$, we perform cross-validation to determine the optimal $\alpha$ for each city and the corresponding $\gamma_t$ at each time interval, with the objective being minimizing the Kolmogorov–Smirnov (KS) statistics between the pdfs of the generated and simulated MCNs for both unweighted and weighted degree distributions (see SI~\sref{SI:parameter_validation}). To validate the correctness of our generation model, we conduct two sample KS tests to compare the cdfs of unweighted and weighted degree distributions between generated and simulated MCNs, with the null hypothesis being that two data samples for comparison are drawn from the same continuous distribution.

The validation results are summarized in Tab~\ref{tab:ksstat}, and we also visualize the fitting of the generated MCNs in Fig.~\ref{fig:curve_fit}. We observe that for all experiments, we fail to reject the null hypothesis for the two sample K-S test with the lowest p-value among these cases being 0.742. Even this lowest value is way above the significant threshold for rejecting the null hypothesis (0.05), and in most cases the $p$ value is greater than 0.95 for both weighted and unweighted distributions. The statistics along with the goodness of fit in Figure~\ref{fig:curve_fit} are indicative that the proposed generation function well captures the underlying mechanisms that govern the meetings of passengers and the duration of exposures during their travels in metro systems. More importantly, the validation of the generation model on three cities provides strong evidence for the existence of a universal rule that shapes the contacts among travelers in transit networks.
\begin{figure*}[!h]
    \centering
    \includegraphics[width=\linewidth]{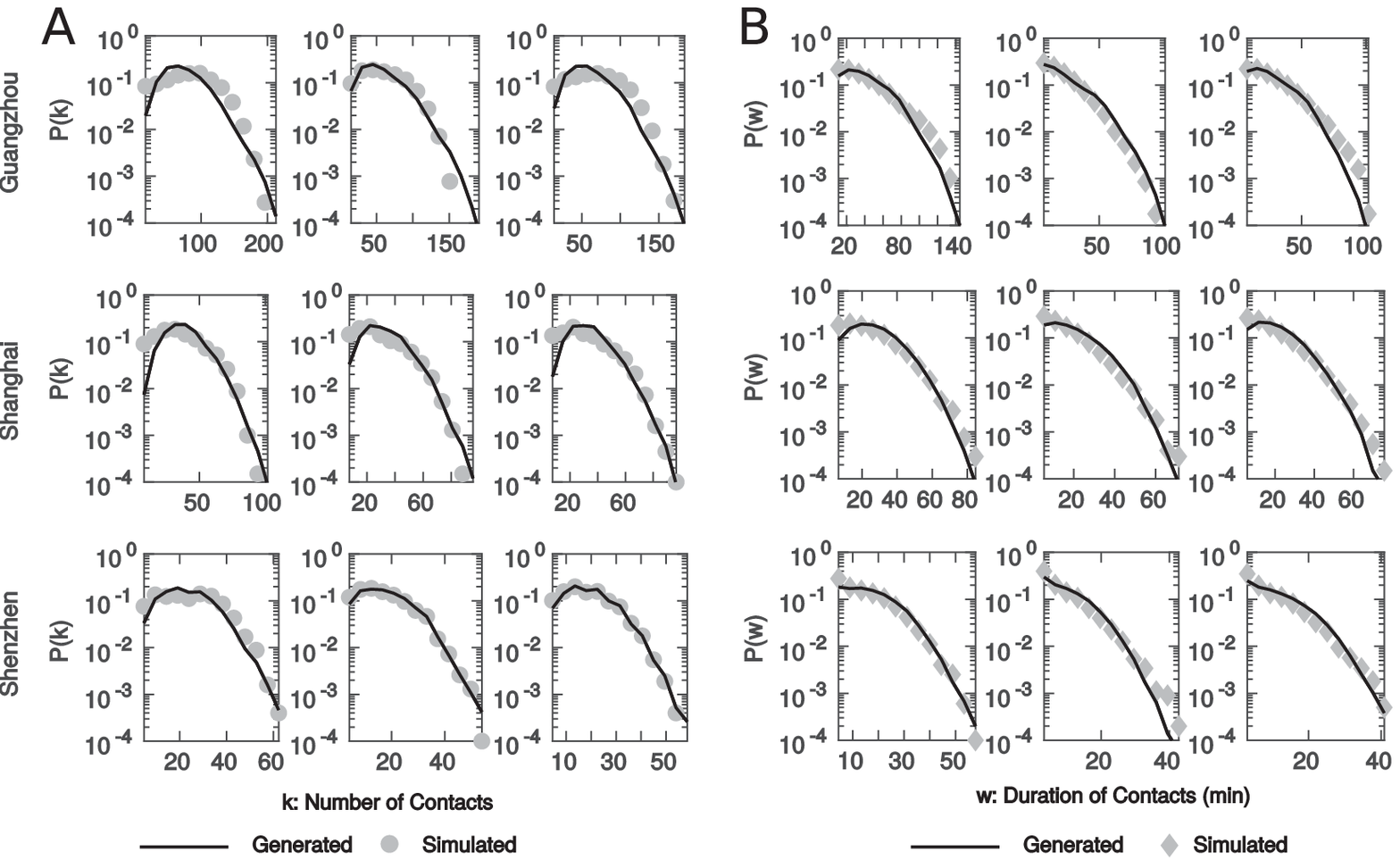}
    \caption{Visualization for the goodness of fit of the generated MCNs as compared to the simulated MCNs from the smart card data. Each row corresponds to the results of the same city and each column corresponds to the results from a particular time of the day. (A) Fitting results of the probability density function for the unweighted degree distributions.  (B) Fitting results of the probability density function for the weighted degree distributions. All results are obtained from the average performance of 50 generated MCNs using the optimal $\gamma_t$ and from the average of 50 simulated MCNs. Each MCN has 1000 nodes. All scenarios fail to reject the null hypothesis of the KS test with very high p-value, where the summary statistics of the KS test and calibrated model parameters are shown in Tab.\ref{tab:ksstat}.}
    \label{fig:curve_fit}
\end{figure*}

\begin{figure}[!h]
    \centering
    \includegraphics[width=.7\linewidth]{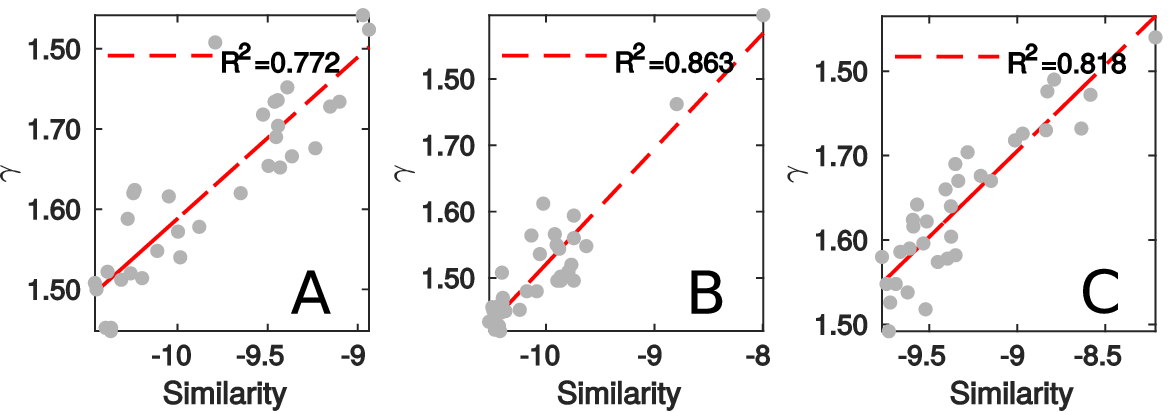}
    \caption{Comparison between the calibrated parameter $\gamma_t$ and trip similarity index for (A) Guangzhou, (B) Shanghai and (C) Shenzhen.}
    \label{fig:validation_similarity}
\end{figure}

\begin{table}[!h]
\centering
\caption{Summary of fitted results and model parameters from KS test for three cities.}
\begin{tabular}{p{3cm}p{1.5cm}p{1.5cm}p{1.5cm}p{1.5cm}}
	\toprule
Time         & 8AM   & 12PM & 6PM     \\
\toprule
\multicolumn{4}{c}{Guangzhou}               \\

unweighted p & 0.742  & 0.742 & 0.742    \\
weighted p   & 0.999  & 0.999 & 0.999    \\
$\gamma_t$        & 0.652 & 0.604  & 0.610  \\
$\alpha$       & \multicolumn{3}{c}{0.005}  \\
\toprule
\multicolumn{4}{c}{Shanghai}                \\

unweighted p & 0.956  & 0.998 & 0.956    \\
weighted p   & 0.999  & 0.999 & 0.998    \\
$\gamma_t$        & 0.634 & 0.604  & 0.616  \\
$\alpha$        & \multicolumn{3}{c}{0.004}     \\
\toprule
\multicolumn{4}{c}{Shenzhen}                \\

unweighted p & 0.956  & 0.999 & 0.999    \\
weighted p   & 0.999  & 0.999 & 0.999    \\
$\gamma_t$        & 0.662 & 0.612  & 0.626  \\
$\alpha$       & \multicolumn{3}{c}{0.005} \\
\bottomrule
\end{tabular}
\label{tab:ksstat}
\end{table}

%% file: discussion.tex
\section*{Discussion}
By inspecting the structure of those simulated MCNs, we observe that there is a lack of high degree nodes. This observation is confirmed by the generation model, which can be decomposed as the weighted combination of Poisson pdfs as shown in equation~\ref{eq:generation_model}. We see that the degree of a node may be drawn from a collection of Poisson distributions with mean $Mw_i$.
This explains the lack of high degree nodes in the contact network as compared to the scale-free network with same number of nodes and average degree, since the probability of having large $k$ in a random network diminishes faster than exponential. On the other hand, the generation model also explains why the tails of the degree distributions of MCNs decay slower than a random network. Note that in MCNs, high degree nodes are generated from the Poisson distribution with a large $Mw_i$ value, which requires longer trip length $t_i$. The decay of the tails is therefore the convolution of the tail of a random network and the tail of the mobility distribution which decays exponentially. In addition, we derive the expressions for the mean $<k>$ and variance ($\sigma^2=<k^2>-<k>^2$) of the MCNs' degree distributions in SI~\sref{SI:first_second_moment}. We find that $<k>\propto N$ and the variance $\sigma^2\propto N^2$, so that both measures diverge with $N\rightarrow \infty$ and the variance is of higher magnitude than the average node degree. This indicates that the degree distributions of MCNs have similar characteristics as compared to the exponential distribution and the number of contacts and the contact duration are bounded by the human mobility in the transit networks. Moreover, these findings are aligned with the empirical observations of $<k>$ and $<k^2>$ in simulated MCNs (Fig.~\ref{SI_fig:deg_diverge}), which further strengthens the validity of the developed generation mechanism. 

One important parameter in the generation model is $\gamma_t$, where we define $\gamma_t$ as the similarity among the trips. To validate this argument, we also quantitatively measure the trip similarity (see SI~\sref{SI:similarity_measures}) among travelers based on the trip OD matrix $Q$. The similarity measure is introduced to quantify the strength of overlapping of a particular trip pair on other trip pairs in terms of contact duration and demand level. We then compare the dominant eigenvalues of $Q$ and we use the variance of the dominant eigenvalues to quantify the similarity of trip purposes. In particular, higher variance suggests that the majority of the trips take place on a few ODs and the trip purposes among these riders are more similar. We compare the computed similarity index with the calibrated $\gamma_t$ and the results are shown in Fig.~\ref{fig:validation_similarity}. We see that the calculated similarity presents a strong linear relationship with $\gamma_t$ among all three cities, with $R^2$ value being above 0.77 if we fit a simple linear function to interpret this relationship. These suggest that similarity can be used as a proxy for $\gamma_t$ for prediction purposes.

While it is difficult to devise effective yet practical control strategies, the degree distribution provides valuable insights in improving the resilience of the transit system by controlling how its contact networks are shaped. To reduce the risk of MCNs, it is equivalent to minimize the probability of the MCNs having high degree nodes. Based on equation~\ref{eq:generation_model}, we know that $p(k)$ is linearly proportional to the number of passengers and the scale of the metro network. Reducing these values will lead to a linear reduction in the average number of contacts while the shape of the degree distribution will remain the same. By observing the metro network layouts, we observe that larger transit network, possibly with more number of lines and transfer stations, may result in lower $\alpha$. But the data used in our study is not sufficient to explain how we may reduce $\alpha$ and this may require further investigation. Alternatively, efforts can be made to reduce $\gamma_t$ so as to sub-linearly decrease the probability of having high degree nodes and result in the degree distribution that decays faster. This can be achieved by segregating passengers through an optimally designed timetable or advising passengers to distribute their departure time. The ultimate solution, however, lies in the distribution of the human mobility distribution for $t_i$. This will not only reshape how the Poisson pdfs are combined but also change the weight of each Poisson distribution. In particular, we would like to pursue the distributions of $t_i$ with faster decaying tails, so that both $Mw_i$ and $p(w_i)$ for larger $w_i$ values will be minimized at the same time. And this can be realized by changing the layout of the metro network or, ultimately, the urban form itself. We can expect that $w_i$ will decay faster by reducing the number of transfers required for the pair of stations of long trip duration, which results in lower maximum trip length and may also contribute to lowering $\alpha$. As for the urban form, a more decentralized urban structure is the most effective way for reducing the risk of communicable diseases, which implies that people can avoid long distance travels across the city as they can find work or entertainment places closer to their home locations. While both approaches are deemed to be effective, the design of the metro network and urban form is not a sole function of the risk of communicable disease. Thus oftentimes we have to compromise among the disease risk, construction cost, efficiency, and also equity of urban mobility. But the developed model in this study provides an important tool for improving the network resilience without undermining other aspects of the system.

\begin{figure}[!h]
    \centering
    \includegraphics[width=.7\linewidth]{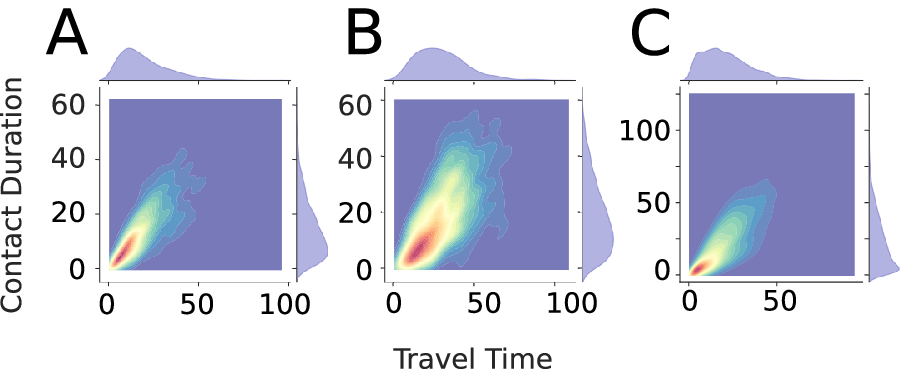}
	\caption{The correlations between travel time and total contagion duration in MCNs of three cities. 50 MCNs are generated for each plot using data from 8:00-8:30 AM, with each MCN having 1000 nodes. (A)-(C) represent results for Guangzhou, Shanghai, and Shenzhen respectively.}
	\label{fig:contact_correlation}
\end{figure}

One final issue is to identify the group of travelers who experience and introduce high-risk exposure in the transit system. To gain insights on this issue, the correlations among the travel time distributions and distributions of contact durations are plotted and shown in Fig.~\ref{fig:contact_correlation}. We see that the travel time and the contact duration are positively correlated and this observation is consistent across all three cities. We can also verify that there is a wide range of travel time for travelers who experience high contact duration in the metro system. In general, the positive correlation suggests that travelers who experience highest contact duration are likely to be those who have the longest travel time. And the travel time of urban commuters is closely related to their work and home locations, their income levels, and eventually their lifestyle and health conditions. One recent study reported that those commuters with the longest travel time in the metro are likely to be low-income migrates, and they may change their home and work locations more frequently than other urban commuters~\cite{huang2018tracking}. This finding implies another potential risk in transit networks. If commuters with long travel time overlap with the low-income population, then these people are likely to be more prone to infection during disease outbreaks. Compared to other population groups, low-income people usually have fewer options (such as time off and sick leaves) and may pay less attention to personal health and hygiene due to limited disposable income~\cite{ettner1996new}. Consequently, the riskiest group of travelers in metro systems are likely to be the most susceptible and vulnerable group of people during the disease outbreak. And this may inevitably raise additional challenges associated with disease contagion and equity of travel in urban transportation networks.